\documentclass[11pt]{article}

\usepackage[utf8]{inputenc}
\usepackage[normalem]{ulem}
\usepackage{enumitem}
\setlist[itemize]{nosep,topsep=3pt,itemsep=3pt}
\setlist[enumerate]{nosep,topsep=3pt,itemsep=3pt}
\usepackage{etoolbox}



\usepackage{amsthm,amsfonts,framed}
\usepackage{mathrsfs}
\usepackage[margin=1in]{geometry}
\usepackage{physics, graphicx}

\usepackage[compact]{titlesec}
\titleformat*{\paragraph}{\bfseries\itshape}

\usepackage{tikz,xcolor,hyperref}
\definecolor{lime}{HTML}{A6CE39}
\DeclareRobustCommand{\orcidicon}{
  \begin{tikzpicture}
  \draw[lime, fill=lime] (0,0) 
  circle [radius=0.16] 
  node[white] {{\fontfamily{qag}\selectfont \tiny ID}};
  \draw[white, fill=white] (-0.0625,0.095) 
  circle [radius=0.007];
  \end{tikzpicture}
  \hspace{-2mm}
}
\foreach \x in {A, ..., Z}{\expandafter\xdef\csname orcid\x\endcsname{\noexpand\href{https://orcid.org/\csname orcidauthor\x\endcsname}
      {\noexpand\orcidicon}}
}

\newcommand{\s}{\hphantom{-}}
\hypersetup{
  colorlinks   = true,    
  urlcolor     = blue,    
  linkcolor    = blue,    
  citecolor    = blue 
}

\usepackage{amsmath,amssymb,amsthm}
\usepackage{mathtools,thmtools}
\usepackage{authblk}
\usepackage{tabu}
\usepackage{enumitem}
\usepackage{dsfont}
\usepackage{braket}
\usepackage{cancel}
\usepackage{graphicx}
\usepackage{xcolor}
\usepackage{empheq}
\usepackage{float}
\usepackage{array, makecell}
\usepackage{esint}
\usepackage{multirow}
\usepackage{setspace}
\usepackage[title]{appendix}
\usepackage{tikz}
\usepackage[capitalize]{cleveref}
\usepackage{algorithm}
\usepackage{algorithmicx, algpseudocode}
\usepackage{comment}
\usepackage{tensor}
\usepackage{multicol}
\usepackage{url}

\setcellgapes{3.5pt}

\def\bra#1{\langle#1|}
\def\ket#1{|#1\rangle}

\newcommand{\av}{{\boldsymbol{a}}}
\newcommand{\bv}{{\boldsymbol{b}}}

\newcommand{\uv}{{\boldsymbol{u}}}
\newcommand{\vv}{{\boldsymbol{v}}}

\newcommand{\zv}{{\boldsymbol{z}}}
\newcommand{\gammav}{{\boldsymbol{\gamma}}}
\newcommand{\Gammav}{{\boldsymbol{\Gamma}}}
\newcommand{\betav}{{\boldsymbol{\beta}}}
\newcommand{\expj}{\mathbb{E}_J}
\newcommand{\vect}[1]{\boldsymbol{#1}}
\newcommand{\paramv}{\gammav,\betav}

\newcommand{\Et}{\tilde{E}}
\newcommand{\fin}{f^{\rm inner}}
\newcommand{\fout}{f^{\rm outer}}
\newcommand{\ain}{\boldsymbol{a}^{\rm inner}}
\newcommand{\aout}{\boldsymbol{a}^{\rm outer}}
\newcommand{\XOR}{\textnormal{XOR}}

\newcommand{\pp}[1]{{[#1]}}
\newcommand{\parent}{\mathfrak{p}}

\newcommand{\EV}{\mathbb{E}}
\newcommand{\qsc}{{[q]}}
\newcommand{\XORSAT}{{XORSAT}}
\newcommand{\conj}[1]{\ddot{#1}}

\newtheorem{thm}{Theorem}
\newtheorem*{thm*}{Theorem}

\newtheorem{lemma}{Lemma}

\newtheorem*{conj*}{Conjecture}

\numberwithin{equation}{section}

\title{The Quantum Approximate Optimization Algorithm at High Depth \\
for MaxCut on Large-Girth Regular Graphs \\
and the Sherrington-Kirkpatrick Model
}

\newcommand{\respace}[1]{\hspace{-5pt}#1\hspace{-1pt}}

\author[1]{\respace\orcidJ{}Joao Basso}
\author[1,2]{{~}Edward Farhi}
\author[3]{\respace\orcidA{}Kunal Marwaha}
\author[1]{\\ \respace\orcidB{}Benjamin Villalonga}
\author[4]{\respace\orcidZ{}Leo Zhou}

\affil[1]{Google Quantum AI, Venice, CA 90291}
\affil[2]{Center for Theoretical Physics, Massachusetts Institute of Technology, Cambridge, MA 02139}
\affil[3]{Department of Computer Science, University of Chicago, Chicago, IL 60637}
\affil[4]{Walter Burke Institute for Theoretical Physics, California Institute of Technology, Pasadena, CA 91125}

\date{July 6, 2022}

\begin{document}

\maketitle

\begin{abstract}
\normalsize
\onehalfspacing
The Quantum Approximate Optimization Algorithm (QAOA) finds approximate solutions to combinatorial optimization problems.
Its performance monotonically improves with its depth $p$.
We apply the QAOA to MaxCut on large-girth $D$-regular graphs. We give an iterative formula to evaluate performance for any $D$ at any depth $p$.
Looking at random $D$-regular graphs, at optimal parameters and as $D$ goes to infinity, we find that the $p=11$ QAOA beats all classical algorithms (known to the authors) that are free of unproven conjectures.
While the iterative formula for these $D$-regular graphs is derived by looking at a single tree subgraph, we prove that it also gives the ensemble-averaged performance of the QAOA on the Sherrington-Kirkpatrick (SK) model defined on the complete graph.
We also generalize our formula to Max-$q$-{\XORSAT} on large-girth regular hypergraphs.
Our iteration is a compact procedure, but its computational complexity grows as $O(p^2 4^p)$. 
This iteration is more efficient than the previous procedure for analyzing QAOA performance on the SK model, and we are able to numerically go to $p=20$.
Encouraged by our findings, we make the optimistic conjecture that the  QAOA, as $p$ goes to infinity, will achieve the Parisi value. 
We analyze the performance of the quantum algorithm, but one needs to run it on a quantum computer to produce a string with the guaranteed performance.
\end{abstract}

\thispagestyle{empty}

\clearpage 
\tableofcontents
\clearpage

\section{Introduction}

We are at the start of an era in which quantum devices are running algorithms.  We need to understand the power of quantum computers for solving or finding approximate solutions to combinatorial optimization problems. One approach is to learn by experimenting on hardware.  Although useful for probing the hardware and testing algorithms at small sizes, it does not give a convincing picture of asymptotic behavior. To this end we need  mathematical studies of the behavior of  quantum algorithms, running on ideal circuits, at large sizes.  In this paper we take a step in that direction by analyzing the Quantum Approximate Optimization Algorithm as applied to a certain combinatorial optimization problem.  The instances are large and the depth of the algorithm is high. For this task, we will see that the QAOA outperforms the best assumption-free classical algorithm.

MaxCut is a combinatorial optimization problem on bit strings whose input is a graph.  Each bit is associated with a vertex, and the goal is to maximize the number of edges with bit assignments that disagree on the two ends of the edge.
It is NP-hard to solve this problem exactly, and even approximating the optimal solution beyond a certain ratio is NP-hard~\cite{trevisan2000gadgets}.
We focus on MaxCut for large-girth $D$-regular graphs.
On these graphs, the currently known best classical algorithms without assuming any unproven conjectures (including Goemans-Williamson and the Gaussian wave process \cite{MontanariSen2016,lyons2017factors,barak2021classical,thompson2021explicit}) achieve an expected cut fraction (the number of cut edges output by the algorithm divided by the number of edges) of $1/2 + (2/\pi)/\sqrt{D}$ as both the girth and $D$ go to $\infty$, where $2/\pi \approx 0.6366$.

We apply the Quantum Approximate Optimization Algorithm (QAOA) \cite{farhi2014quantum} to large-girth $D$-regular graphs. The QAOA depends on a parameter $p$, the algorithm's depth. At small $p$, the QAOA has been realized in current quantum hardware \cite{harrigan2021quantum}. Some analytic results are also known. At $p=1$, the QAOA has a guaranteed approximation ratio (the number of cut edges output by the algorithm divided by the maximum  number of edges that can be cut) of at least $0.6924$ on all $3$-regular graphs \cite{farhi2014quantum} and an expected cut fraction of at least $1/2 + 0.3032/\sqrt{D}$ on triangle-free graphs \cite{wang2018quantum}. For $p=2$, the QAOA has an approximation ratio of at least $0.7559$ on $3$-regular graphs with girth more than $5$ and, for $p=3$, that ratio becomes $0.7924$ when the girth is more than $7$ \cite{wurtz2020bounds}. So far, expressions for the QAOA's performance on any fixed-$D$ regular,  large-girth graph are known only for $p=1$ \cite{wang2018quantum} and $p=2$ \cite{marwaha2021local}.

In this work, we analyze the performance of the QAOA on any large-girth $D$-regular graph for any choice of $p$ by looking at a single tree subgraph.
Using the regularity of this tree subgraph, we derive an iteration that computes the performance of the QAOA.
After optimizing over the $2p$ input parameters, we find that the $p=11$ QAOA improves on $1/2 + (2/\pi)/\sqrt{D}$, when $D$ is large and the girth is more than $23$.
This is better than all assumption-free classical algorithms known to the authors.%
\footnote{There is a recent classical message-passing algorithm~\cite{alaoui2021local} that also does better than $1/2+(2/\pi)/\sqrt{D}$ for MaxCut on large-girth $D$-regular graphs. It gets asymptotically close to the optimum assuming the solution space has no ``overlap gap property'' (see \cite{gamarnik2021overlap} for a review).
}

We also show that this performance, obtained from one subgraph, is mathematically equal to the ensemble-averaged performance of the QAOA applied to the Sherrington-Kirkpatrick (SK) model~\cite{farhi2019quantum}.
This implies that the iteration in this paper can also be used to give the QAOA's performance on the SK model.
A recent related work can be found in Ref.~\cite{boulebnane2021predicting}.
Our iteration is more efficient than the one originally shown in Ref.~\cite{farhi2019quantum}, and we have been able to go numerically to higher depth. 

Encouraged by our findings, we conjecture that the large $p$ performance of the QAOA will achieve the optimal cut fraction on large random $D$-regular graphs, where a vanishing fraction of neighborhoods are not locally tree-like. The optimal cut fraction on these graphs is also related to the SK model. It is $1/2 + \Pi_*/\sqrt{D} + o(1/\sqrt{D})$, where $\Pi_* = 0.763166\ldots$, the Parisi value, is the ground state energy density of the SK model \cite{parisi1979toward,dembo2017extremal}. If our conjecture is right we have a simple, though computationally intensive, new iteration for calculating the Parisi value $\Pi_*$.

Generalizing our formalism, we also analyze the performance of the QAOA for Max-$q$-XORSAT (of which MaxCut is a special case at $q=2$) on large-girth $D$-regular hypergraphs.
The $p=1$ QAOA was recently found to do better than an analogous classical threshold algorithm for $q > 4$~\cite{marwaha2021bounds}.
The iterative formula for general $q$ is very similar to that for MaxCut and has the same time and memory complexities in the $D\to\infty$ limit.
We run this iteration to find optimal QAOA parameters and performance for $3 \leq q \leq 6$ and $1 \leq p \leq 14$. Moreover, we discuss potential obstructions to the QAOA from not ``seeing'' the whole graph.


\section{Background on the QAOA and MaxCut}
\label{sec:background}

The QAOA \cite{farhi2014quantum} is a quantum algorithm for finding approximate solutions to combinatorial optimization problems. The cost function counts the number of clauses satisfied by an input string.   Given a cost function $C(\zv)$ on strings $\zv \in \{\pm 1\}^n$,  we can define a corresponding quantum operator, diagonal in the computational basis, as $C\ket{\zv} = C(\zv) \ket{\zv}$.
Moreover, let $B = \sum_{j=1}^n X_j$, where $X_j$ is the Pauli $X$ operator acting on qubit $j$.
Let $\vect\gamma = (\gamma_1,\gamma_2,\ldots,\gamma_p)$ and $\vect\beta = (\beta_1,\beta_2,\ldots,\beta_p)$. 
The QAOA initializes the system of qubits in the state $\ket{s} = \ket{+}^{\otimes n}$ and applies $p$ alternating layers of $e^{-i\gamma_j C}$ and $e^{-i\beta_j B}$ to prepare the state
\begin{equation} \label{eq:wavefunction}
\ket{\paramv} = e^{-i\beta_p B} e^{-i\gamma_p C} \cdots e^{-i\beta_1 B} e^{-i\gamma_1 C} \ket{s}.
\end{equation}
For a given cost function $C$, the corresponding QAOA objective function is
$\braket{\paramv|C|\paramv}$.
Preparing the quantum state $\ket{\paramv}$ and then measuring in the computational basis enough times, one will find a bit string $\zv$ such that $C(\zv)$ is near $\braket{\paramv|C|\paramv}$ or better.

We study the performance of the QAOA on MaxCut. Given a graph $G = (V, E)$ with vertices in $V$ and edges in $E$, the MaxCut cost function is
\begin{equation} \label{eq:MaxCut}
    C_{\text{MC}}(\zv) = \sum_{(u,v) \in E} \frac{1}{2}(1 - z_u z_v) .
\end{equation}
We restrict our attention to graphs that are regular and have girth greater than $2p+1$.
We work with these graphs because the subgraph that the QAOA at depth $p$ sees on them are regular trees and this enables our calculation.  Here, by ``seeing'' we refer to the fact that the output of the QAOA on a qubit depends only on a neighborhood of qubits that are within distance $p$ to the given qubit on the graph. 
In what follows, we focus on ($D+1$)-regular graphs, which implies the subgraph seen by the QAOA on each edge is a $D$-ary tree.

With $D$ large, we will see that the optimal $\gammav$ are of order $1/\sqrt{D}$. So we find it convenient to prepare the QAOA state $\ket{\paramv}$ using the scaled cost function operator
\begin{equation}
    C = -\frac{1}{\sqrt{D}} \sum_{(u,v) \in E} Z_u Z_v ,
\end{equation}
where we have subtracted a constant that only introduces an irrelevant phase. The factor of 1/2 has been dropped so that this form of the cost function will match the cost function used in the Sherrington-Kirkpatrick model.
Note we are preparing the state $\ket{\paramv}$ using $C$ as a driver instead of the $C_{\rm MC}$ operator.
With this scaling, the optimal $\gammav$ will be of order unity instead of $1/\sqrt{D}$.

Given any edge in a $(D+1)$-regular graph with girth greater than $2p+1$ the subgraph with vertices at most $p$ away from the edge is a $D$-ary tree regardless of which edge.
Since the QAOA at depth $p$ only sees these trees, we have
\begin{equation}
    \braket{\paramv | C_{\text{MC}} | \paramv} = \frac{1}{2} |E| \Big(1 - \braket{ \paramv | Z_u Z_v | \paramv } \Big)
\end{equation}
where $(u,v) \in E$ is any edge. The cut fraction output by the QAOA is then
\begin{equation}
     \frac{\braket{\paramv | C_{\text{MC}} | \paramv}}{|E|} = \frac{1}{2} - \frac{1}{2} \braket{ \paramv | Z_u Z_v | \paramv }.
\end{equation}
Since the QAOA cannot beat the optimal cut fraction of $1/2 + \text{order}( 1/\sqrt{D})$ in a typical random regular graph, we write 
\begin{equation}\label{eq:implicit_def_eta}
    \frac{1}{2}\braket{\paramv | Z_u Z_v | \paramv} = -   \frac{\nu_p(D,\paramv)}{\sqrt{D}} 
\end{equation}
where $\nu_p(D,\paramv)$ for good parameters will be of order unity.

\vspace{10pt}

\section{The QAOA on large-girth $(D+1)$-regular graphs}\label{sec:QAOA_tree_iterations}

We describe two iterations to evaluate the performance of the QAOA at high depth on MaxCut on large-girth $(D+1)$-regular graphs.
The cut fraction output by the QAOA at any parameters is
\begin{equation} \label{eq:MC-and-nu}
    \frac{\braket{\paramv | C_{\text{MC}} | \paramv}}{|E|} = \frac{1}{2} + \frac{\nu_p(D,\paramv)}{\sqrt{D}}.
\end{equation}
We give one iteration to evaluate $\nu_p(D,\paramv)$ at finite $D$, and one for the $D \to \infty$ limit.
We have attempted to make this section self-contained for those readers only interested in the form of the iterations, and deferred the detailed proofs of these iterations to \cref{sec:iter_proofs}.

In what follows, we index vectors in the following order:
\vspace{-5pt}
\begin{equation}
    \av = (a_1, a_2, \cdots_, a_p, a_0, a_{-p}, \cdots, a_{-2}, a_{-1}) \,.
\end{equation}
Define, for $1 \leq r \leq p$,
\begin{align} \label{eq:Gamma-def}
    \Gamma_{r} = \gamma_r, \qquad
    \Gamma_0 = 0, \qquad
    \Gamma_{-r} = -\gamma_r.
\end{align}
That is, $\Gammav$ is a $(2p+1)$-component vector.
Furthermore, let
\begin{align}\label{eq:f_definition}
f(\av) &= \frac{1}{2} \braket{a_1 | e^{i \beta_1 X} | a_2} \cdots \braket{ a_{p-1} | e^{i \beta_{p-1} X} | a_p} \braket{a_p | e^{i \beta_p X} | a_0} \nonumber \\
    &\quad \times \braket{a_{0} | e^{-i \beta_p X} | a_{-p}} \braket{a_{-p} | e^{-i \beta_{p-1} X} | a_{-(p-1)}} \cdots \braket{a_{-2} | e^{-i \beta_1 X} | a_{-1}}
\end{align}
where $a_i \in \{+1, -1\}$ enumerates the two computational basis states, and
\vspace{-5pt}
\begin{equation}
\braket{a_1|e^{i\beta X} |a_2} = 
\begin{cases}
\cos(\beta) &\text{ if } a_1 = a_2 \\
i\sin(\beta) &\text{ if } a_1 \neq a_2.
\end{cases}
\vspace{-5pt}
\end{equation}

\subsection{An iteration for any finite $D$}\label{sec:finite_D_iteration}
Here we give an iteration that allows us to evaluate $\nu_p(D,\paramv)$ for any input parameters and $D$.

Let $H_D^{(m)} \colon \{-1, 1\}^{2p+1} \rightarrow \mathbb{C}$ for $0 \leq m \leq p$ with 
\begin{equation}
H_D^{(0)}(\av) = 1
\end{equation}
and 
\begin{equation}\label{eq:H_m_def}
    H_D^{(m)}(\av) = \bigg(\sum_{\bv} f(\bv) H_D^{(m-1)}(\bv) \cos{\Big[ {\textstyle \frac{1}{\sqrt{D}}}\Gammav \cdot (\av \bv) \Big]} \bigg)^D \quad \text{ for } 1 \leq m \leq p 
\end{equation}
where we denote $\av\bv$ as the entry-wise product, i.e. $(\av\bv)
_j = a_j b_j$.
By starting with $H_D^{(0)}(\av)=1$ and iteratively evaluating \cref{eq:H_m_def} for $m=1,2,\ldots,p$, we arrive at $H_D^{(p)}(\av)$ that can be used to compute
\begin{align}\label{eq:ZL_ZR_finite_D}
    \nu_p(D,\paramv) = {\textstyle \frac{i \sqrt{D}}{2}} \sum_{\av, \bv} & a_0 b_0 f(\av) f(\bv) H_D^{(p)}(\av) H_D^{(p)}(\bv) \sin{\Big[ {\textstyle \frac{1}{\sqrt{D}}} \Gammav \cdot (\av \bv) \Big]}.
\end{align}
We prove this in \cref{sec:iter_proofs_finite}.

Note that each step of the above iteration involves a sum with $2^{2p+1}$ terms for each of the $2^{2p+1}$ entries of $H_D^{(m)}(\av)$.
The final step has a sum with $O(16^p)$ terms. Overall, this iteration has a time complexity of $O(p\, 16^p)$ and a memory complexity of $O(4^p)$.
This is much faster than the original ``light cone'' approach that directly evaluates $\braket{Z_u Z_v}$ on the subgraph seen by the QAOA~\cite{farhi2014quantum}. That procedure takes $2^{O({D^p})}$ time without utilizing the symmetric structure of the regular tree subgraph.

\subsection{An iteration for $D \rightarrow \infty$}\label{sec:infinite_D_iteration}

We find that in the infinite $D$ limit we get a more compact iteration which takes fewer steps to evaluate. We state the result here and prove it in \cref{sec:iter_proofs_infinite}.

Define matrices $G^{(m)} \in \mathbb{C}^{(2p+1) \times (2p+1)}$ for $0 \leq m \leq p$ as follows.
For $j,k \in \{1,\dots,p, 0, -p,\dots$, $-1\}$, let 
\begin{equation}
G_{j,k}^{(0)} = \sum_{\av} f(\av) a_j a_k
\end{equation}
and
\begin{equation}\label{eq:inf_D_iter_G_m}
    G_{j,k}^{(m)} = \sum_{\av} f(\av) a_j a_k \exp
    \Big({-}\frac{1}{2} \sum_{j',k'=-p}^p G_{j',k'}^{(m-1)} \Gamma_{j'} \Gamma_{k'} a_{j'} a_{k'} 
    \Big)
    \quad
    \text{for } 1 \leq m \leq p.
\end{equation}
Starting at $m=0$ and going up by $p$ steps, we arrive at $G^{(p)}$ which is used to compute
\begin{equation}\label{eq:infinite_D_iter_result}
     \nu_p(\paramv) := \lim_{D \rightarrow \infty} \nu_p(D,\paramv) = \frac{i}{2} \sum_{j = -p}^p \Gamma_j ( G^{(p)}_{0, j} )^2.
\end{equation}
Since there are $p+1$ matrices with $O(p^2)$ entries, and each involves a sum over $O(4^p)$ terms, this iteration na\"ively has a time complexity of $O(p^3   4^p)$. This is quadratically better than the time complexity of the finite-$D$ formula.
The memory complexity is only $O(p^2)$ for storing the $G^{(m)}$ matrix, which is exponentially better than $O(4^p)$ memory needed to store the entries of $H_D^{(m)}$ in the finite-$D$ iteration.

We note some properties about this iteration.
Superficially Eq.~\eqref{eq:inf_D_iter_G_m} looks like a recursive map on the matrices $G^{(m)}$ which one might think would only asymptotically converge in the number of steps.
However it converges to a fixed point $G^{(p)}$ after $p$ steps in a highly structured way.
In particular, the iteration has the following three sets of properties, which we prove in Appendix~\ref{apx:iter-properties}.
We use the convention $1 \leq r < s \leq p$ and $j, k \in \left\{ 1, \ldots, p, 0, -p, \ldots, -1 \right\}$.

\begin{enumerate}[label=(\alph*)]
\item Values of the diagonal and anti-diagonal of $G^{(m)}$ are all 1.
$G^{(m)}$ is symmetric with respect to the diagonal, reflection with respect to the anti-diagonal results in complex conjugation, and the matrix consists of 8 triangular regions which are rotations, reflections, and/or complex conjugations of each other.
To be precise, $G^{(m)}$ satisfies the following properties:
    \begin{multicols}{2}
    \begin{enumerate}[label=(\arabic*)]
        \item  $G_{j,k}^{(m)} = G_{k,j}^{(m)}$
        \item  $G_{j,j}^{(m)} = G_{j,-j}^{(m)} = 1$
        \item  $G_{0, r}^{(m)} = G_{0,-r}^{(m)*}$
        \item  $G_{r,s}^{(m)} = G_{r,-s}^{(m)}= G_{-r,-s}^{(m)*} = G_{-r,s}^{(m)*}$
    \end{enumerate}
    \end{multicols}
These are sketched in  \cref{fig:G_matrix_structure}.

\item \label{item:dependence} $G^{(m)}_{r, s}$ only depends on $G^{(m-1)}_{r', s'}$ where $1 \leq r' < s'< s$.
Similarly, $G_{0,r}^{(m)}$ only depends on $G_{r',s'}^{(m-1)}$ for $1\le r' < s' \le p$.

\item As a consequence of \ref{item:dependence}, at each step $m$ of the iteration the corner blocks of size $(m + 1) \times (m + 1)$ of $G^{(m)}$ converge to their final value, i.e., they reach a fixed point and do not change in later iteration steps.
This implies that matrix $G^{(p)}$ is a fixed point.
This is sketched in \cref{fig:G_matrix_structure}, where matrix entries of the same color reach their fixed point at the same step of the iteration, starting from the corners and ending with the central ``cross'' at step $p$.
\end{enumerate}

Making use of \ref{item:dependence} and some properties of $f(\av)$ allows us to lower the complexity of the iterative procedure to $O(p^2 4^p)$.
We show this in Appendix~\ref{apx:proof_placement_G}.

\begin{figure}[h]
    \centering
    \includegraphics[height=3.5in]{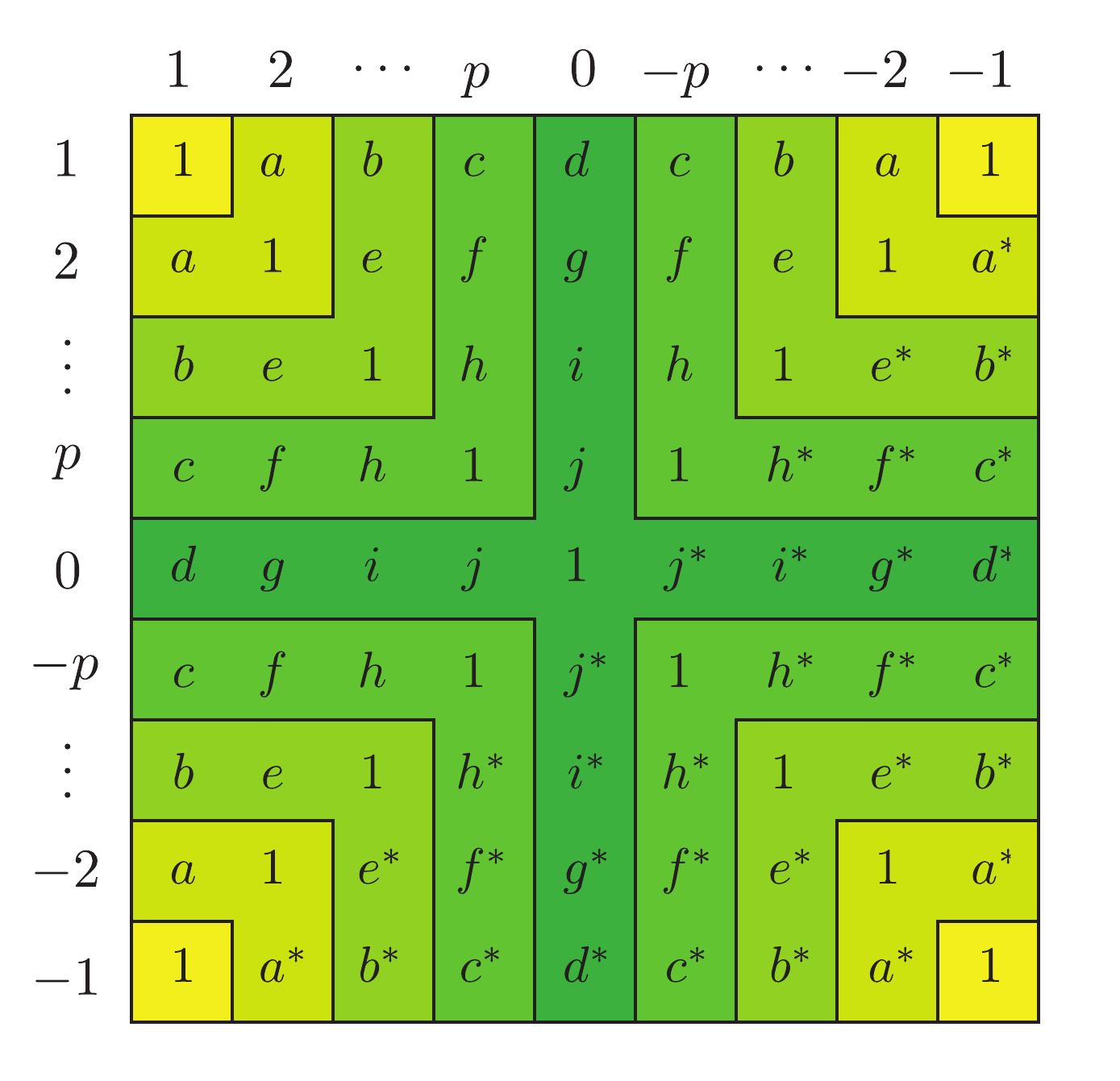}
    \vspace{-10pt}
    \caption{Sketch of the properties of matrices $G^{(m)}$ in the iterative formula of Section~\ref{sec:infinite_D_iteration}, at $p=4$.
    Regions of the same color converge in the same iteration step, starting from the corners and with the central row and column converging after $p$ steps.}
    \label{fig:G_matrix_structure}
\end{figure}


\section{Proof of the iterations}\label{sec:iter_proofs}
In this section, we prove the correctness of the two iterations in \cref{sec:finite_D_iteration} and \cref{sec:infinite_D_iteration}.
These proofs illustrate two key technical ideas in this paper: namely, we can exploit the regularity of the tree subgraph seen by the QAOA to yield a compact formula for its performance, and we find an algebraic simplification in the $D\to\infty$ limit.
The reader not interested in the techniques can skip to \cref{sec:numerics}.

\subsection{Proof of the finite $D$ iteration}
\label{sec:iter_proofs_finite}
We now prove the finite $D$ iteration that was stated in \cref{sec:finite_D_iteration}.
We focus on the iteration for $p=2$ as an example, and its generalization to other $p$ is immediate.

The goal is to evaluate the energy expectation for a single edge $(L, R)$ on a $(D+1)$-regular graph whose girth is larger than $2p+1$.
For $p=2$, this is
\begin{align} \label{eq:ZLZR-p=2}
    \braket{\paramv | Z_L Z_R | \paramv} =
    \bra{s} e^{i\gamma_1 C} e^{i\beta_1 B} e^{i \gamma_2 C} e^{i\beta_2 B} Z_L Z_R e^{-i\beta_2 B} e^{-i \gamma_2 C} e^{-i\beta_1 B} e^{-i\gamma_1 C} \ket{s}
\end{align}
where $C=-(1/\sqrt{D}) \sum_{(u,v)\in E} Z_u Z_v$, and $E$ denotes the set of edges for the given graph.
In the Heisenberg picture, it can be seen that the operator $e^{i\gamma_1 C} \cdots e^{i\beta_p B} Z_L Z_R e^{-i\beta_p B} \cdots e^{-i\gamma_1 C}$ only acts nontrivially on the subgraph induced by including all vertices distance $p$ or less from either node $L$ or $R$.
For a $(D+1)$-regular graph with girth greater than $2p+1$, this subgraph looks like a pair of $D$-ary trees that are glued at their roots (see \cref{fig:tree_graph}), with a total of $n=2(D^p + \cdots +D+1)$ nodes.
In what follows, we compute \cref{eq:ZLZR-p=2} by restricting our attention to only the qubits in this subgraph.

\begin{figure}[t]
\centering
    \includegraphics[height=2.65in]{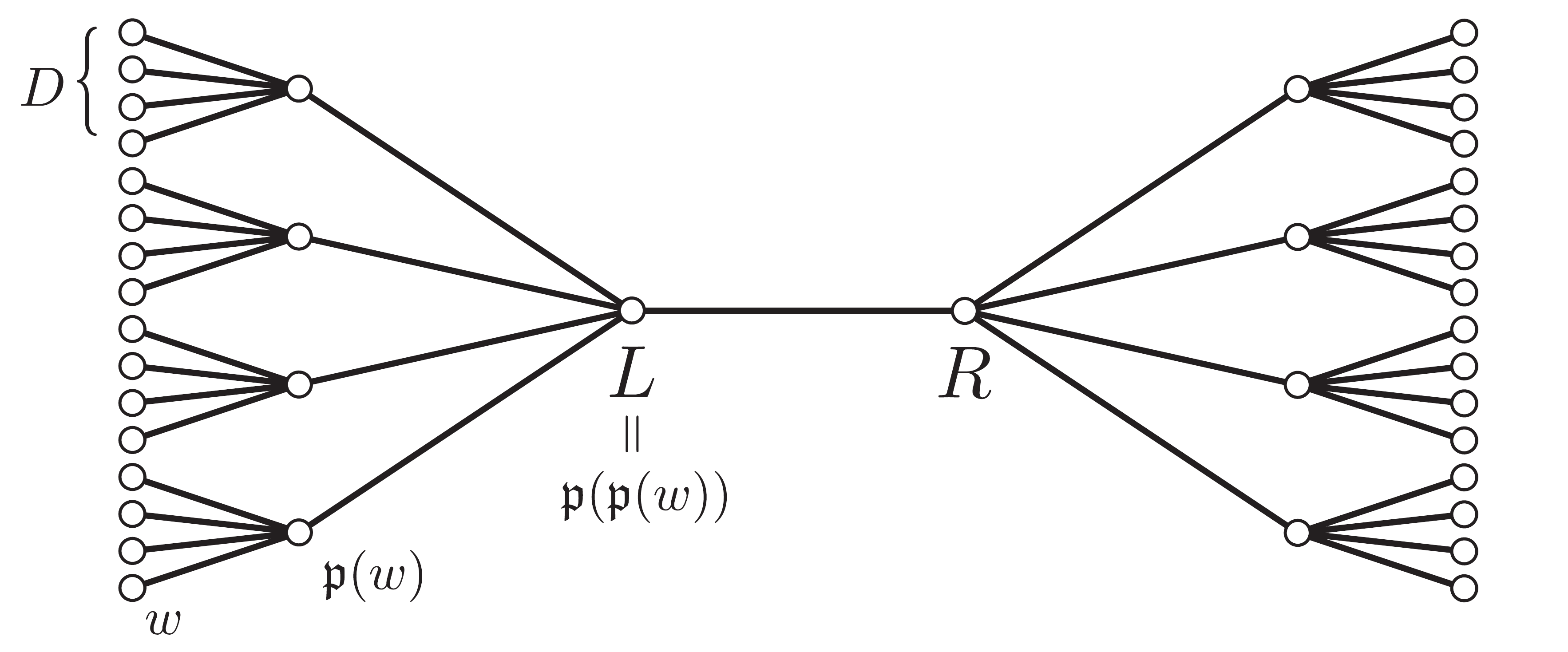}
\caption{The tree subgraph seen by the QAOA at $p=2$ for the edge $(L, R)$ on a $(D+1)$-regular graph with girth $> 2p+1$.
For any node $v$ on either of the $D$-ary trees we denote $\parent(v)$ as the parent of that node. In the figure $w$ is a leaf node, and we show its parent and its parent's parent.}
\vspace{-8pt}
\label{fig:tree_graph}
\end{figure}

We start by inserting 5 complete sets in the computational $Z$-basis that we will label as $\zv^\pp{1}, \zv^\pp{2}, \zv^\pp{0}, \zv^\pp{-2}$, and $\zv^\pp{-1}$. Each of these complete sets iterates over $2^n$ basis states since the number of qubits in the subgraph is $n$. Then
\begin{align}
&\braket{\paramv| Z_L Z_R|\paramv } =
    \sum_{\{\zv^\pp{i} \}} \braket{s |\zv^\pp{1}} 
    e^{i\gamma_1 C(\zv^\pp{1})}
    \braket{\zv^\pp{1}| e^{i\beta_1 B} |\zv^\pp{2}}
    e^{i\gamma_2 C(\zv^\pp{2})}
    \braket{\zv^\pp{2}| e^{i\beta_2 B} |\zv^\pp{0}} z^\pp{0}_L z^\pp{0}_R \nonumber \\
&\qquad\qquad\qquad \qquad
\quad \times
    \braket{\zv^\pp{0}| e^{-i\beta_2 B} |\zv^\pp{-2}}
    e^{-i\gamma_2 C(\zv^\pp{-2})}
    \braket{\zv^\pp{-2}| e^{-i\beta_1 B} |\zv^\pp{-1}}
    e^{-i\gamma_1 C(\zv^\pp{-1})} \braket{\zv^\pp{-1}|s} \nonumber \\
& \qquad\qquad = \frac{1}{2^n} \sum_{\{\zv^\pp{i} \}}
    \exp\Big[ i\gamma_1 C(\zv^\pp{1}) + i\gamma_2 C(\zv^\pp{2}) - i\gamma_2 C(\zv^\pp{-2})- i\gamma_1 C(\zv^\pp{-1}) \Big] 
    z^\pp{0}_L z^\pp{0}_R \nonumber \\
& \qquad \qquad \quad \times
    \prod_{v=1}^n \braket{z_v^\pp{1}|e^{i\beta_1 X}|z_v^\pp{2}}
     \braket{z_v^\pp{2}|e^{i\beta_2 X}|z_v^\pp{0}}
      \braket{z_v^\pp{0}|e^{-i\beta_2 X}|z_v^\pp{-2}}
       \braket{z_v^\pp{-2}|e^{-i\beta_1 X}|z_v^\pp{-1}}.
\label{eq:ZLZR-deriv-1}
\end{align}
Let us define the following function which is the $p=2$ version of Eq.~\eqref{eq:f_definition}:
\enlargethispage{\baselineskip}
\begin{equation} \label{eq:fdef-p=2}
f(a_1, a_2, a_0, a_{-2}, a_{-1}) = 
\frac{1}{2}
\braket{a_{1}|e^{i\beta_1 X}|a_{2}}
     \braket{a_{2}|e^{i\beta_2 X}|a_{0}}
      \braket{a_{0}|e^{-i\beta_2 X}|a_{-2}}
       \braket{a_{-2}|e^{-i\beta_1 X}|a_{-1}}.
\end{equation}
Then, using $\Gammav$ as defined in Eq.~\eqref{eq:Gamma-def}, we can rewrite Eq.~\eqref{eq:ZLZR-deriv-1} as 
\begin{equation}
\braket{\paramv| Z_L Z_R|\paramv } = 
\sum_{\{\zv^\pp{i} \}}
z^\pp{0}_L z^\pp{0}_R
 \exp\Big[ i\sum_{j=-2}^2 \Gamma_j C(\zv^\pp{j})\Big]  \prod_{v=1}^n f(\zv_v)
\end{equation}
where $\zv_v = (z_v^\pp{1}, z_v^\pp{2}, z_v^\pp{0}, z_v^\pp{-2}, z_v^\pp{-1})$ are the bits from the 5 complete sets associated with node $v$.
Using the fact that $C(\zv) = -(1/\sqrt{D}) \sum_{(u,v) \in E} z_u z_v$, we can rewrite $\braket{\paramv| Z_L Z_R|\paramv }$ as
\begin{align}
\braket{\paramv| Z_L Z_R|\paramv } &= 
\sum_{\{\zv_u \}}
z^\pp{0}_L z^\pp{0}_R
 \exp\Big[ {-} {\textstyle \frac{i}{\sqrt{D}}}  \sum_{(u',v')\in E} \Gammav  \cdot (\zv_{u'} \zv_{v'}) \Big]  \prod_{v=1}^n f(\zv_v)
 \label{eq:ZLZR-deriv-2}
\end{align}
where we have replaced the sum over the $2p+1$ complete sets $\{\zv^\pp{i}: -2\le i \le 2\}$ with an equivalent sum over the bit configurations of each node $\{\zv_u: 1 \le u \le n\}$.
Now to evaluate $\braket{Z_LZ_R}$ we need to perform a sum over the bit configurations $\zv_v$ of every node $v$ in the tree subgraph, where each node is coupled to its neighbors on the graph via the term in the exponential of \cref{eq:ZLZR-deriv-2}.

We can start by considering a single leaf node $w$ who is only connected to its parent node
$\parent(w)$ on the tree, as shown in \cref{fig:tree_graph}. 
Then the sum over the 32 bit values of the configuration $\zv_w= (z_w^\pp{1}, z_w^\pp{2}, z_w^\pp{0}, z_w^\pp{-2}, z_w^\pp{-1})$ yields
\vspace{-5pt}%
\begin{equation}
\sum_{\zv_w} f(\zv_w) \exp\Big[ {-} {\textstyle \frac{i}{\sqrt{D}}}  \Gammav  \cdot (\zv_w \zv_{\parent(w)})  \Big] 
\end{equation}
which is a function of the parent node's configuration $\zv_{\parent(w)}$.
Note that doing this on every leaf node contributes the same function to its parent.
Since there are exactly $D$ leaf nodes per parent, we get the following contribution
\vspace{-5pt}%
\begin{equation}
H^{(1)}_D(\zv_{\parent(w)}) := \bigg(
\sum_{\zv_w} f(\zv_w) \exp\Big[ {-} {\textstyle \frac{i}{\sqrt{D}}}  \Gammav  \cdot (\zv_w \zv_{\parent(w)})  \Big]
\bigg)^D.
\end{equation}
This is true for every parent node of any of the leaves.

After performing the sums for all the leaf nodes, we can move to the sums for their parents.
Let us look at the sum on the node $\parent(w)$ for example, which yields
\vspace{-3pt}
\begin{align}
\sum_{\zv_{\parent(w)}} f(\zv_{\parent(w)}) H^{(1)}_D(\zv_{\parent(w)}) \exp\Big[ {-} {\textstyle \frac{i}{\sqrt{D}}} \Gammav  \cdot (\zv_{\parent(w)} \zv_{\parent(\parent(w))})  \Big].
\end{align}
Again, because its parent node $\parent(\parent(w))$ has $D$ identical children like $\parent(w)$, this yields
\begin{align}  \label{eq:ppw}
H_D^{(2)}(\zv_{\parent(\parent(w))}) := \bigg(
\sum_{\zv_{\parent(w)}} f(\zv_{\parent(w)}) H^{(1)}_D(\zv_{\parent(w)}) \exp\Big[ {-} {\textstyle \frac{i}{\sqrt{D}}}  \Gammav  \cdot (\zv_{\parent(w)} \zv_{\parent(\parent(w))})  \Big]
\bigg)^D.
\end{align}
Note at $p=2$ we have reached the root of the tree $L=\parent(\parent(w))$ after these two iterations.

To evaluate $\braket{\paramv| Z_L Z_R|\paramv }$, it only remains to sum over the $5$ bits in $\zv_L$ and the $5$ bits in $\zv_R$:
\vspace{-3pt}
\begin{align}
\braket{\paramv| Z_L Z_R|\paramv } = \sum_{\zv_L, \zv_R} 
z^\pp{0}_L z^\pp{0}_R f(\zv_L) f(\zv_R) H_D^{(2)}(\zv_L) H_D^{(2)}(\zv_R)
 \exp\Big[ {-} {\textstyle \frac{i}{\sqrt{D}}} \Gammav  \cdot (\zv_L \zv_R) \Big] .
\end{align}

\enlargethispage{\baselineskip}

For higher $p$, we can see that the evaluation of $\braket{\paramv| Z_L Z_R|\paramv }$ simply involves more iterations of Eq.~\eqref{eq:ppw} corresponding to more levels in the tree subgraph.
In summary, the iteration for general $p$ can be written as starting with
$H^{(0)}_D(\av) = 1$
and then evaluating for $m=1,2,\ldots, p$,
\vspace{-3pt}
\begin{align} \label{eq:iter-with-exp}
H^{(m)}_D(\av) = \bigg(\sum_\bv f(\bv) H^{(m-1)}_D(\bv) \exp\Big[ {-} {\textstyle \frac{i}{\sqrt{D}}} \Gammav \cdot (\av\bv)\Big] \bigg)^D,
\end{align}
since there are $p$ levels in the tree subgraph seen by the QAOA with $p$ layers.
At the end we get
\begin{align} \label{eq:ZLZR-with-exp}
\braket{\paramv| Z_L Z_R|\paramv } = \sum_{\av, \bv} & a_0b_0 f(\av) f(\bv) H_D^{(p)}(\av) H_D^{(p)}(\bv) \exp\Big[{-}{\textstyle \frac{i}{\sqrt{D}}} \Gammav \cdot (\av \bv) \Big].
\end{align}
This is almost what we have stated for the iteration in \cref{sec:finite_D_iteration}.

To finish the proof, we note from Eq.~\eqref{eq:fdef-p=2} as well as its general $p$ version in Eq.~\eqref{eq:f_definition} that
\begin{equation}
f(-\av) = f(\av).
\end{equation}
We now claim that
\begin{equation} \label{eq:Hm-negation-identity}
    H_D^{(m)}(-\av) = H_D^{(m)}(\av) \quad \text{for } 0 \leq m \leq p
\end{equation}
which we will show by induction on $m$. Note this is trivially true for the base case $m=0$ since $H_D^{(0)}(\av)=1$ is constant. Assuming that $H_D^{(m-1)}(-\av) = H_D^{(m-1)}(\av)$, we can take $\bv \to -\bv$ in the summand of \cref{eq:iter-with-exp} and combine it with its original form to see that
\begin{equation} \label{eq:HDm-final}
H^{(m)}_D(\av) 
= \bigg(\sum_\bv f(\bv) H^{(m-1)}_D(\bv) \cos\Big[ {\textstyle \frac{1}{\sqrt{D}}} \Gammav \cdot (\av\bv)\Big]\bigg)^D.
\vspace{-5pt}
\end{equation}
From this form it follows that $H_D^{(m)}(-\av) = H_D^{(m)}(\av)$ since $\av$ only appears in the cosine which is an even function, establishing Eq.~\eqref{eq:Hm-negation-identity}.

Similarly, we can take $\bv\to-\bv$ in Eq.~\eqref{eq:ZLZR-with-exp} and combine with its original form to get
\begin{equation}
\braket{\paramv| Z_L Z_R|\paramv } 
 = -i\sum_{\av, \bv} 
a_0 b_0 f(\av) f(\bv) H_D^{(p)}(\av) H_D^{(p)}(\bv)
 \sin\Big[ {\textstyle \frac{1}{\sqrt{D}}} \Gammav  \cdot (\av \bv) \Big] .
\vspace{-5pt}%
\end{equation}
Thus to get the $\nu_p$ as defined in \cref{eq:implicit_def_eta} that tells us the cut fraction, we have
\begin{equation}
\nu_p(D, \paramv)
= \frac{i \sqrt{D}}{2} \sum_{\av, \bv}  a_0 b_0 f(\av) f(\bv) H_D^{(p)}(\av) H_D^{(p)}(\bv) \sin{\Big[ {\textstyle \frac{1}{\sqrt{D}}} \Gammav \cdot (\av \bv) \Big]}.
\end{equation}
This proves our iteration for any finite $D$ in Section~\ref{sec:finite_D_iteration}.

\vspace{8pt}
\subsection{Proof of $D \rightarrow \infty$ iteration}
\label{sec:iter_proofs_infinite}
We wish to evaluate \cref{eq:ZL_ZR_finite_D} in the $D \rightarrow \infty$ limit:
\begin{equation}\label{eq:ZL_ZR_finite_D_2nd_time}
    \lim_{D \rightarrow \infty} \nu_p(D,\paramv) 
= \lim_{D \rightarrow \infty} \frac{i \sqrt{D}}{2} \sum_{\av, \bv}  a_0 b_0 f(\av) f(\bv) H_D^{(p)}(\av) H_D^{(p)}(\bv) \sin{\Big[ {\textstyle \frac{1}{\sqrt{D}}} \Gammav \cdot (\av \bv) \Big]}.
\end{equation}

We first prove by induction that for $0 \leq m \leq p$,
\begin{equation}
    H^{(m)}(\av) := \lim_{D \rightarrow \infty} H_D^{(m)}(\av)
\end{equation}
exists and is finite. For $m=0$, our claim holds because $H_D^{(0)}(\av) = 1$. Assuming the claim is true for $m-1$, we examine $H^{(m)}(\av)$ by taking the limit on
\cref{eq:HDm-final}
\begin{equation}
    H^{(m)}(\av) = \lim_{D \rightarrow \infty} \bigg[\sum_{\bv} f(\bv) H_D^{(m-1)}(\bv) \cos\Big( {\textstyle \frac{1}{\sqrt{D}}} \Gammav \cdot (\av \bv) \Big)\bigg]^D.
\end{equation}
Then performing a Taylor expansion of $\cos(\cdots)$, we get
\begin{equation}
     H^{(m)}(\av) = \lim_{D \rightarrow \infty} \bigg[ \sum_{\bv} f(\bv) H_D^{(m-1)}(\bv) \bigg(1 - {\textstyle \frac{1}{2D}} \Big(\Gammav \cdot (\av \bv) \Big)^2 + O\Big( {\textstyle \frac{1}{D^2}} \Big) \bigg)  \bigg]^D.
\end{equation}

Using the fact that for any $m$,
\begin{equation}\label{eq:sum_f_H_in_main_text}
    \sum_\av f(\av) H_D^{(m)}(\av) = 1
\end{equation}
which is proved as Lemma~\ref{lem:fH-sum} in \cref{apx:H_properties}, we get
\begin{equation}
     H^{(m)}(\av) = \lim_{D \rightarrow \infty} \Big[ 1 - {\textstyle \frac{1}{2D}} \sum_{\bv} f(\bv) H_D^{(m-1)}(\bv) \big(\Gammav \cdot (\av \bv) \big)^2 + O\big( {\textstyle \frac{1}{D^2}} \big) \Big]^D .
\end{equation}

Finally, taking the limit,
\begin{equation}\label{eq:H_tilde_m}
    H^{(m)}(\av) = \exp \Big[ -\frac{1}{2} \sum_\bv f(\bv) H^{(m-1)}(\bv) \big(\Gammav \cdot (\av \bv) \big)^2 \Big]
\end{equation}
which yields an iteration on $H^{(m)}$.

Returning to \cref{eq:ZL_ZR_finite_D_2nd_time}, we apply the product rule of limits to $H_D^{(p)}(\av), H_D^{(p)}(\bv)$, and $\sqrt{D} \sin[\Gammav\cdot(\av\bv)/\sqrt{D}]$ and get
\begin{equation}\label{eq:inf_D_proof_with_H_tilde}
    \lim_{D \rightarrow \infty} \nu_p(D,\paramv) = \frac{i}{2} \sum_{\av, \bv} a_0 b_0 f(\av) f(\bv) H^{(p)}(\av) H^{(p)}(\bv) \Gammav \cdot (\av \bv). 
\end{equation}

This iteration can be simplified by expanding the dot products in Eqs.~\eqref{eq:H_tilde_m} and \eqref{eq:inf_D_proof_with_H_tilde} to get
\begin{align}
H^{(m)}(\av)
    &= \exp\Big[ {-} \frac{1}{2}  \sum_{j,k=-p}^p \Gamma_j \Gamma_k a_j a_k \Big(\sum_\bv f(\bv) H^{(m-1)}(\bv) b_j b_k\Big)\Big] \label{eq:H_tilde_m_grouped}, \\
\lim_{D \rightarrow \infty} \nu_p(D,\paramv) 
    &= \frac{i}{2} \sum_{j=-p}^p \Gamma_j \Big(\sum_{\av} f(\av) H^{(p)}(\av) a_0 a_j \Big) \Big(\sum_{\bv}  f(\bv) H^{(p)}(\bv) b_0 b_j \Big) \label{eq:lim_eta_grouped}
\end{align}
and noticing that the quantity $\sum_{\av} f(\av) H^{(m)}(\av) a_j a_k$ appears repeatedly. For $0\leq m \leq p$ and $ -p \le j,k \le p$, define
\begin{equation}\label{eq:G_definition}
    G_{j,k}^{(m)} := \sum_{\av} f(\av) H^{(m)}(\av) a_j a_k.
\end{equation}
For $m=0$, this is 
\begin{equation}
    G^{(0)}_{j,k} = \sum_{\av} f(\av) a_j a_k.
\end{equation}
For $1\le m \le p$, we plug \cref{eq:H_tilde_m_grouped} into \cref{eq:G_definition} to get
\begin{align}\label{eq:G_element_explicit}
    G_{j,k}^{(m)} = \sum_{\av} f(\av) a_j a_k \exp\Big[{ -\frac{1}{2} \sum_{j',k'=-p}^p G^{(m-1)}_{j',k'}  \Gamma_{j'} \Gamma_{k'} a_{j'} a_{k'}  }\Big].
\end{align}

Finally, \cref{eq:lim_eta_grouped} can be written as 
\begin{align}
    \lim_{D \rightarrow \infty} \nu_p(D,\paramv)= \frac{i}{2} \sum_{j = -p}^p \Gamma_j (G^{(p)}_{0, j})^2
\end{align}
which establishes the iteration stated in \cref{sec:infinite_D_iteration}.

\clearpage
\section{Numerical evaluation and optimization}
\label{sec:numerics}

Let
\begin{equation}\label{eq:nu_p}
    \nu_p(\paramv) = \lim_{D \rightarrow \infty} \nu_p(D,\paramv).
\end{equation}
Numerically implementing the iteration summarized in \cref{sec:infinite_D_iteration} and optimizing for $\gammav, \betav$ we find
\begin{equation}\label{eq:bar_nu_p}
    \bar\nu_p = \max_{\gammav, \betav} \nu_p(\gammav, \betav)
\end{equation}
up to $p=17$.
The values are given in Table~\ref{table:optimal_values} and plotted in \cref{fig:nu_bar} as a function of $1/p$.
The optimal $\gammav$ and $\betav$ can be found in Table~\ref{table:optimal_params} in \cref{apx:optimal_angles}, and some examples are plotted in \cref{fig:params}.
Based on the smooth pattern of the optimal $\gammav$ and $\betav$ up to $p$ of 17, we guess these parameters at $p=18, 19, 20$ using heuristics similar to that in Ref.~\cite{ZhouQAOA}.
Then evaluation of $\nu_p(\paramv)$ gives lower bounds on $\bar{\nu}_p$ at higher $p$ which are listed in Table~\ref{table:lower_bound_values}, and their corresponding $\gammav$ and $\betav$ are listed in Table~\ref{table:lower_bound_params}.

\begin{figure}[ht]
    \centering
    \includegraphics[width=\linewidth]{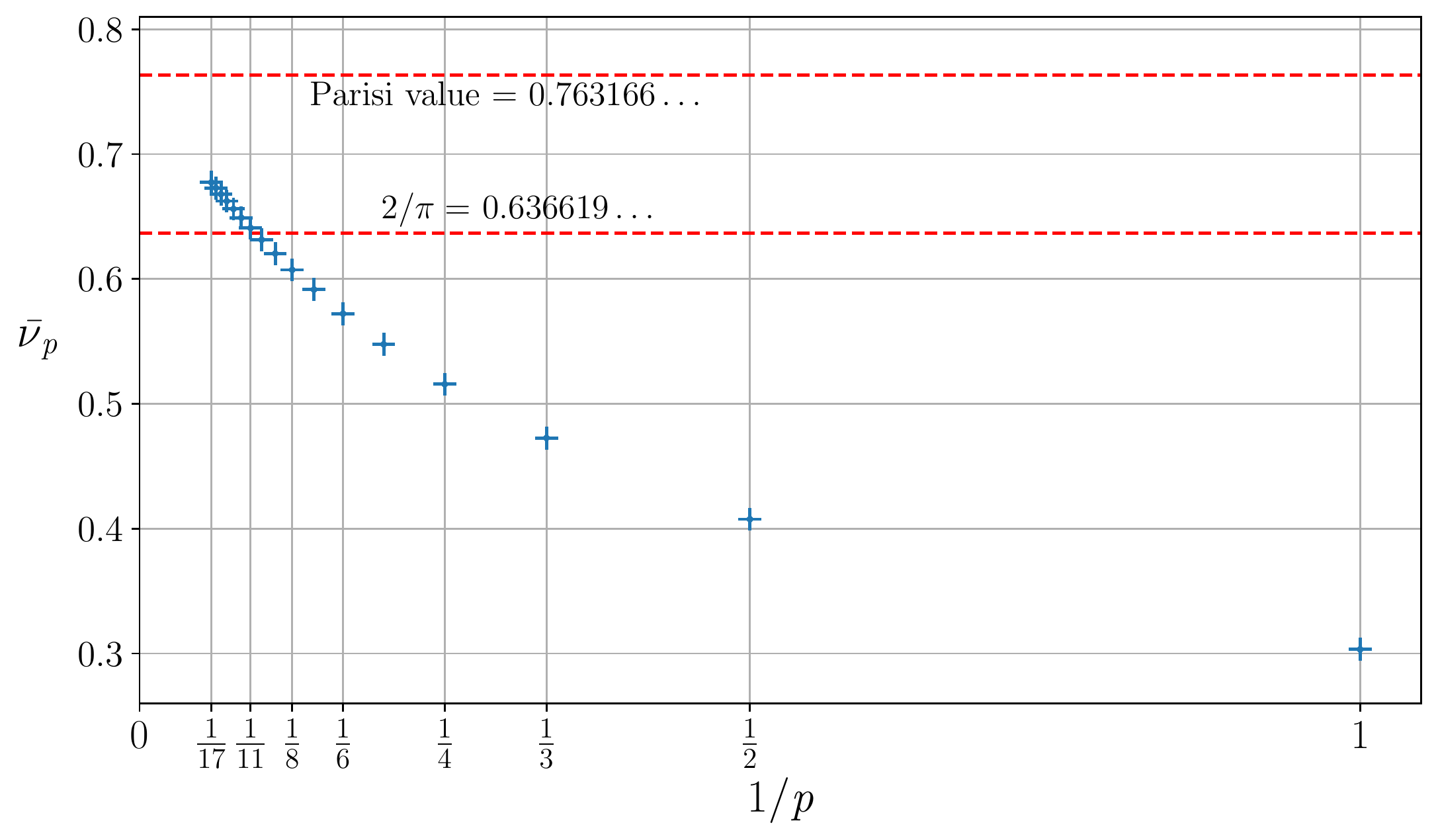}
    \caption{
	Optimal values $\bar{\nu}_p$ as a function of $1 / p$.
    At $p = 11$, $\bar{\nu}_p$ exceeds $2 / \pi$, related to the cut fraction of the best currently known assumption-free classical algorithms. 
    Here we made the somewhat arbitrary choice of plotting the data against $1 / p$ to see the large $p$ region in a compact plot.
    }
    \label{fig:nu_bar}
\end{figure}

\begin{figure}[bth]
    \centering
    \includegraphics[width=\linewidth]{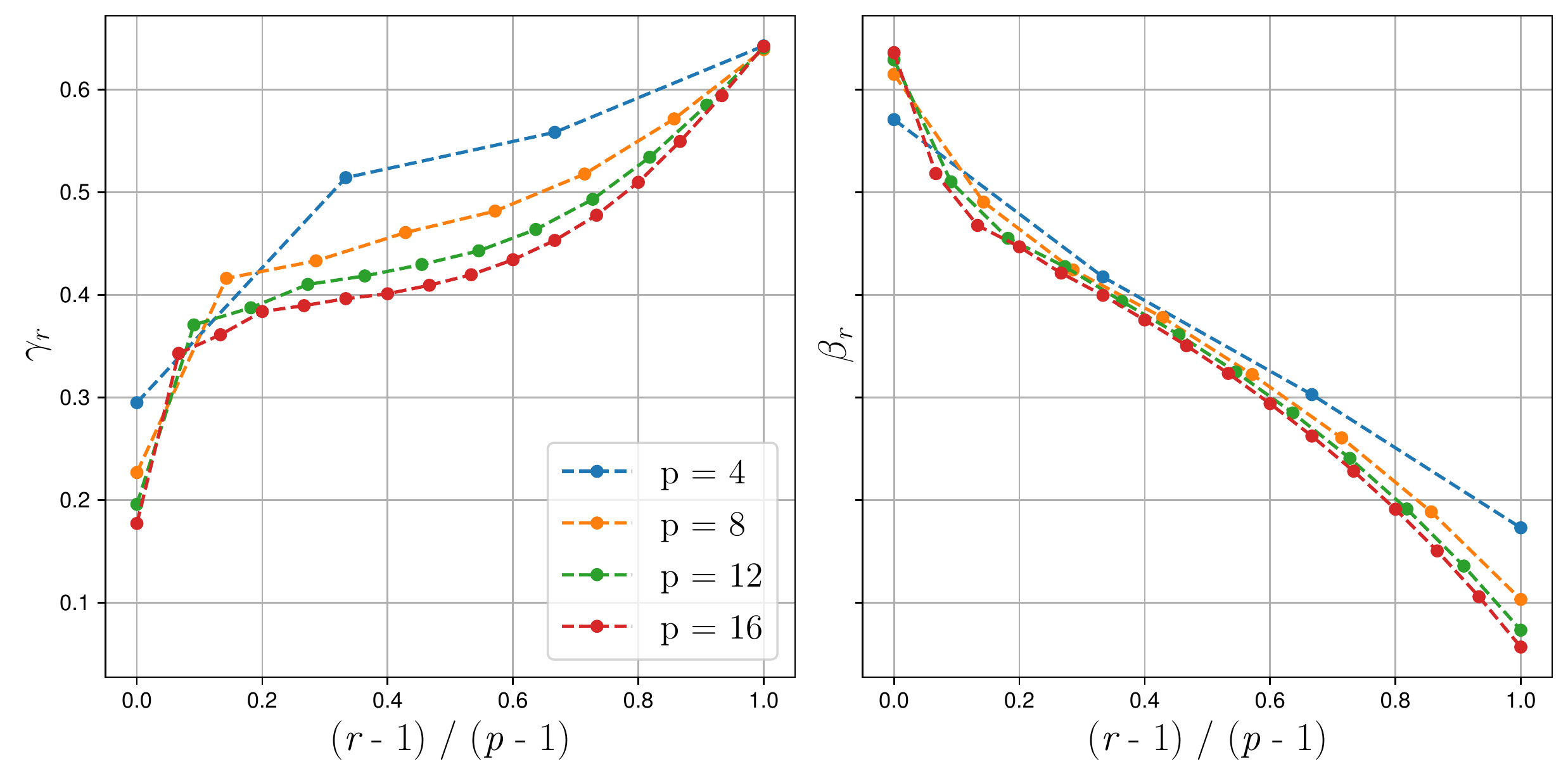}
    \caption{
    Optimal $\gamma_r$ and $\beta_r$ as a function of $(r - 1) / (p - 1) \in [0, 1]$ for $p = 5, 9, 13, 17$.
    For each $p$, the index $r = 1,2,\ldots,p$ enumerates the entries of $\gammav$ and $\betav$.
    Dashed lines in between data points are solely intended to guide the eye.
    }
    \label{fig:params}
\end{figure}
Note that, at $p=11$ and beyond, the QAOA achieves a cut fraction better than $\frac{1}{2} + \frac{2/\pi}{\sqrt{D}}$ in the large $D$ limit, making it the best currently known assumption-free algorithm for MaxCut on large random regular graphs.

We implement the iterative procedure described in \cref{sec:infinite_D_iteration} in C++.
Our code is available at Ref.~\cite{data}.
Bit strings are encoded as \texttt{unsigned long int} variables, which allow for fast bit-wise manipulations.
Matrices and vectors are implemented using the Eigen library~\cite{eigenweb}.
We parallelize the sum over $\av$ in Eq.~\eqref{eq:inf_D_iter_G_m} using OpenMP~\cite{dagum1998openmp}.
We optimize $\paramv$ for each value of $p$ using the LBFGS++ library, which implements the Limited-memory BFGS algorithm for unconstrained optimization problems~\cite{lbfgspp}.
Each evaluation of the gradient of $\nu_p(\paramv)$ is a subroutine of the optimization which takes $2p + 1$ function calls.
We run on a \texttt{n2d-highcpu-224} machine in Google Cloud, which has 224 vCPUs, using one thread per vCPU.
A function call at $p = 16$ takes about 133 seconds, and a function call at $p = 17$ takes about 595 seconds.
The run time of each function call is roughly multiplied by 4 every time $p$ is increased by 1.
At $p = 20$, a single function call takes slightly under 14 hours to evaluate.
Memory usage is dominated by the need to store matrix $G^{(m)}$, which is negligible and quadratic in $p$.
Further optimizations might be possible.

\begin{table}[ht]
\centering
\makegapedcells
\begin{tabular}{cccccccccc}
\hline
\multicolumn{1}{|c|}{$p$} & 1 & 2 & 3 & 4 & 5 & 6 & 7 & 8 & \multicolumn{1}{c|}{9} \\ \hline
\multicolumn{1}{|c|}{$\bar{\nu}_p$} & $0.3033$ & $0.4075$ & $0.4726$ & $0.5157$ & $0.5476$ & $0.5721$ & $0.5915$ & $0.6073$ & \multicolumn{1}{c|}{$0.6203$} \\ \hline 
\\
\cline{1-9}
\multicolumn{1}{|c|}{$p$} & 10 & 11 & 12 & 13 & 14 & 15 & 16 & \multicolumn{1}{c|}{17} \\ \cline{1-9}
\multicolumn{1}{|c|}{$\bar{\nu}_p$} &  $0.6314$ & $0.6408$ & $0.6490$ & $0.6561$ & $0.6623$ & $0.6679$ & $0.6729$ & \multicolumn{1}{c|}{$0.6773$} \\ \cline{1-9}
\end{tabular}
\caption{Optimal values of $\bar{\nu}_p$ up to $p=17$.}
\label{table:optimal_values}
\end{table}

\begin{table}[htb]
\centering
\makegapedcells
\begin{tabular}{cccc}
\hline
\multicolumn{1}{|c|}{$p$} & 18 & 19 & \multicolumn{1}{c|}{20} \\
\hline
\multicolumn{1}{|c|}{$\bar{\nu}_p$ lower bound} & $0.6813$ & $0.6848$ & \multicolumn{1}{c|}{$0.6879$} \\ \hline 
\end{tabular}
\caption{\label{table:lower_bound_values}
Lower bounds of $\bar{\nu}_p$ for $p=18, 19, 20$.}
\end{table}

\clearpage

\newcommand{\lo}[1]{\hat{#1}}

\section{Agreement with the Sherrington-Kirkpatrick model}\label{sec:SK_tree_agreement}

We note that Table~\ref{table:optimal_values} in this paper seems to be an extension of Table 1 in Ref.~\cite{farhi2019quantum}. There, the authors study the performance of the QAOA on the Sherrington-Kirkpatrick (SK) model~\cite{panchenko2013sherrington}, which describes a spin-glass system with all-to-all random couplings.
The cost function is
\begin{equation}
    C_J^{\rm SK}(\zv) = \frac{1}{\sqrt{n}} \sum_{1 \leq i < j \leq n} J_{ij} z_i z_j
\end{equation}
where the $J_{ij}$ are independently drawn from a distribution with mean $0$ and variance $1$. The authors arrive at an iterative formula for the ensemble-averaged performance of the QAOA on the SK model
\begin{equation} \label{eq:Vp-def}
    V_p(\gammav, \betav) := \lim_{n \rightarrow \infty} \expj \Big[ \tensor[_{J}]{\braket{\paramv | C_J^{\rm SK}/n | \paramv }}{_{J}} \Big],
\end{equation}
where $\ket{\paramv}_J$ is the QAOA state prepared with $C_J^{\rm SK}$.
Since concentration is shown to hold, we know that typical instances of the SK model all behave as the ensemble average.

Observe that $\bar\nu_p$, the optimized values of $\nu_p(\paramv)$, listed in \cref{table:optimal_values} of this paper agree with the values of $\bar V_p = \max_{\paramv} V_p(\paramv)$ in Table 1 of Ref.~\cite{farhi2019quantum}.
It turns out that this is true in a general sense:
\begin{thm} \label{thm:SKequiv}
For all $p$ and all parameters $(\paramv)$, we have
\begin{equation}
    V_p(\paramv) = \nu_p(\paramv).
\end{equation}
\end{thm}
This theorem establishes the fact that for each $p$ and fixed parameters, the performance of the QAOA on large-girth $D$-regular graphs in the $D\to\infty$ limit is {\it equal} to its performance on the SK model in the $n\to\infty$ limit.
We remark that in the iteration in this paper there is only one tree subgraph, with of order $D^p$ vertices, for every large-girth $D$-regular graph.
On the other hand, in the SK case, there is an ensemble of instances given by different weights on the complete graph.
It is interesting to us that the ensemble average in Eq.~\eqref{eq:Vp-def} can be replaced by a single subgraph.

Theorem~\ref{thm:SKequiv} also implies that the iteration in \cref{sec:infinite_D_iteration} works for evaluating the performance of the QAOA applied to both large-girth regular graphs and the SK model.

In the rest of this section, we will prove this theorem by showing that the infinite-$D$ iteration in Section~\ref{sec:infinite_D_iteration} of this paper is equivalent to the iteration for the SK model in Section 4 of Ref.~\cite{farhi2019quantum}.  
Showing this equivalence is a bit
cumbersome since the two iterations are written with different conventions.  We have attempted to keep the current section self-contained so the reader can skip this proof and still read the rest of the paper.

\paragraph{Proof of \cref{thm:SKequiv}.}---
To help bridge the two approaches, let
\begin{align}
A = \{(a_1,\ldots a_p, a_{-p}, \ldots, a_{-1}) : a_i = \pm 1\}    
\end{align}
be the set of $2p$-bit strings that are used in the derivations of Ref.~\cite{farhi2019quantum}.
And let 
\begin{align}
B = \{(a_1,\ldots a_p, a_0, a_{-p}, \ldots, a_{-1}) : a_i = \pm 1\}    
\end{align}
be the set of $(2p+1)$-bit strings that are used in the derivations of the current paper.

We start with the version of the infinite-$D$ iteration of Eq.~\eqref{eq:H_tilde_m} which we restate here:
\begin{equation} \label{eq:Hm-again}
    H^{(m)}(\av) = \exp \Big[ -\frac{1}{2} \sum_{\bv\in B} f(\bv) H^{(m-1)}(\bv) \big(\Gammav \cdot (\av \bv) \big)^2 \Big] \quad \text{for }1\le  m \le p, 
\end{equation}
where $H^{(0)}(\av)=1$.
Note importantly that $H^{(m)}(\av)$ does not depend on $a_0$ for $0\le m \le p$ since $\Gamma_0=0$.
Hence, in what follows, we will slightly abuse notation to write the function $H^{(m)}(\av)$ to take either an argument $\av \in A$ or $\av \in B$.

We will now match this iteration to the one in Ref.~\cite{farhi2019quantum} for evaluating $V_p(\paramv)$ that uses symbols such as $Q_\av$, $\Phi_{\av}$, $\Delta_{\av,\bv}$, and $W_\av$, which we will define as we progress in this proof.

For any $(2p+1)$-bit string $\av \in B$, let us define the $2p$-bit string $\lo\av \in A$ via $\lo{a}_{\pm {r}} = a_{\pm r} a_{\pm(r+1)}$ for $1\le r\le p-1$, and $\lo{a}_{\pm p} = a_{\pm p} a_0$. More explicitly,
\begin{align}  \label{eq:lo-a-def}
\lo{a}_1 &= a_1 a_2,  &\ldots,&  
&\lo{a}_{p-1} &= a_{p-1} a_p,   &\lo{a}_p &= a_p a_0, \nonumber \\
\lo{a}_{-1} &= a_{-1} a_{-2}, &\ldots,&  
&\lo{a}_{-(p-1)} &= a_{-(p-1)} a_{-p},   &\lo{a}_{-p} &= a_{-p} a_0 .
\end{align}
Then using the fact that $\braket{a_1 |e^{i\beta X} |a_2} = \braket{a_1 a_2 | e^{i\beta X} |1}$, we can rewrite $f(\av)$ defined in \cref{eq:f_definition}  for any $\av\in B$ as
\begin{align} \label{eq:f-and-Q}
f(\av) &= \frac{1}{2} \braket{\lo{a}_1 | e^{i \beta_1 X} | 1} \cdots \braket{ \lo{a}_{p-1} | e^{i \beta_{p-1} X} | 1} \braket{\lo{a}_p | e^{i \beta_p B } |1} \nonumber \\
    &\quad \times \braket{1 | e^{-i \beta_p X} | \lo{a}_{-p}} \braket{1 | e^{-i \beta_{p-1} X} | \lo{a}_{-(p-1)}} \cdots \braket{1 | e^{-i \beta_1 X} | \lo{a}_{-1}} \nonumber \\
    &= \frac12 Q_{\lo\av}
\end{align}
where $Q_{\av}$ is defined in Section 4 of Ref.~\cite{farhi2019quantum}:
\begin{equation}
Q_{{\av}} := \prod_{j=1}^p
	(\cos\beta_j)^{1 + ({a}_j + {a}_{-j})/2}
	(\sin\beta_j)^{1 -({a}_j + {a}_{-j})/2}
	(i)^{({a}_{-j} - {a}_j) /2}
	 \qquad \text{for any }\av \in A.
\end{equation}
Now, Ref.~\cite{farhi2019quantum} makes use of the following $*$ operation defined on the set $A$ of $2p$-bit strings  as
\begin{equation}\label{eq:star_def}
a^*_r=a_r a_{r+1}\cdots a_p
\qquad \text{and} \qquad
a^*_{-r} = a_{-r} a_{-r-1}\cdots a_{-p}
\qquad
\text{for } 1\le r\le p \,.
\end{equation}
Please take care to note that in this proof, $*$ is used only for the above operation and not complex conjugation.
Furthermore, Ref.~\cite{farhi2019quantum} defines $\Phi_\av$ as
\begin{equation}
    \Phi_\av := \sum_{r=1}^p \gamma_r (a^*_r - a^*_{-r}) \qquad \text{for any }\av \in A.
\end{equation}
Note $\lo{a}^*_{\pm r} = \lo{a}_{\pm r} \cdots \lo{a}_{\pm p}  = a_{\pm r} a_0$.
Then for any $\av\in B$, $\Phi_{\lo\av}$ can be written as
\begin{equation} \label{eq:Phi-and-Gamma}
\Phi_{\lo\av} = \sum_{r=1}^p \gamma_r (\lo{a}^*_r - \lo{a}^*_{-r}) = \sum_{r=1}^p \gamma_r (a_r - a_{-r}) a_0 = (\Gammav \cdot \av) a_0 .
\end{equation}
Since $a_0^2=1$, we get $(\Gammav\cdot \av)^2 = \Phi_{\lo\av}^2$.
Also note $\widehat{\av\bv} = \lo\av\lo\bv$, so we can rewrite \cref{eq:Hm-again} as
\begin{equation} \label{eq:Hm-Q-with-half}
    H^{(m)}(\av) = \exp \Big[ {-} \frac{1}{2} \sum_{\bv\in B} \frac12 Q_{\lo\bv} H^{(m-1)}(\bv) \Phi_{\lo\av \lo\bv}^2 \Big].
\end{equation}
Note that
\begin{equation} \label{eq:b-hat-star}
    H^{(m)}(\lo\bv^*) = H^{(m)}(\bv b_0)
\qquad \text{ for any } \bv\in B,
\end{equation}
where we slightly abuse notation to allow $H^{(m)}$ to take two types of argument: $\lo\bv^*\in A$ and $\bv\in B$.
This equality follows from the fact that $H^{(m)}(\av)$ for $\av\in B$ does not depend on $a_0$, the 0-th component of $\av$.
Since $H^{(m)}(-\bv)=H^{(m)}(\bv)$ which we have shown in Eq.~\eqref{eq:Hm-negation-identity}, we have
\begin{equation} \label{eq:Hm-lo-star}
H^{(m)}(\lo\bv^*) =  H^{(m)}(\bv)
\qquad \text{ for any } \bv\in B.
\end{equation}
Hence, in Eq.~\eqref{eq:Hm-Q-with-half} we can sum over $\lo\bv\in A$ instead of $\bv \in B$, killing a 1/2 factor from the redundancy of the sum over $b_0$. Also we can replace $H^{(m)}(\av) = H^{(m)}(\lo\av^*)$ and write Eq.~\eqref{eq:Hm-Q-with-half} as
\begin{equation}
H^{(m)}(\lo\av^*) = \exp \Big[ {-} \frac12\sum_{\lo\bv\in A} Q_{\lo\bv} H^{(m-1)}(\lo\bv^*) \Phi_{\lo\av \lo\bv}^2 \Big] \qquad \text{ for any } \lo\av \in A.
\end{equation}
We can then drop the hats and rewrite this as
\begin{equation} \label{eq:Hm-Q-Phi}
H^{(m)}(\av^*) = \exp \Big[ {-} \frac12\sum_{\bv\in A} Q_{\bv} H^{(m-1)}(\bv^*) \Phi_{\av \bv}^2 \Big] 
\qquad \text{ for any } \av \in A.
\end{equation}

Now let us define for $0\le m\le  p$ and any ${}\av\in A$
\begin{equation}
R^{(m)}_{{}\av} := Q_{{}\av} H^{(m)}({}\av^*) .
\end{equation}
Then we have $R^{(0)}_{{}\av} = Q_{{}\av}$, and plugging Eq.~\eqref{eq:Hm-Q-Phi} into the above yields
\begin{equation} \label{eq:R-iter}
R^{(m)}_{{}\av} 
= Q_{{}\av} \exp\Big( {-} \frac12\sum_{{}\bv\in A} R_{{}\bv}^{(m-1)} \Phi_{{}\av {}\bv}^2 \Big).
\end{equation}

So far, we have transformed the iteration \eqref{eq:Hm-again} on $H^{(m)}(\av)$ for $\av\in B$ to the above iteration on $R_\av^{(m)}$ for $\av \in A$.
This looks very similar as the iterative formula that yields $W_\av$ in Section 4 of Ref.~\cite{farhi2019quantum}, which is then used to give $V_p(\paramv)$.
To show they are the same, i.e., $R_\av^{(p)}=W_\av$, we need to describe a bit more of the formalism in Ref.~\cite{farhi2019quantum}.
There, the authors define a subset $A_{p+1} \subset A$ where
\begin{equation}
    A_{p+1} = \{ \av : {a}_j = {a}_{-j}\}.
\end{equation}
Ref.~\cite{farhi2019quantum} has also defined a one-to-one ``bar'' operation that takes any ${}\av \not\in A_{p+1}$ to $\bar{{}\av}\not\in A_{p+1}$ such that $Q_{\bar{{}\av}} = - Q_{{}\av}$. This operation is its own inverse.
Furthermore, we have the following fact:
\begin{equation} \label{eq:H-bar-prop}
H^{(m)}({}\av^*) = 1 \quad \text{ if  } \quad {}\av \in A_{p+1} \qquad
\text{and} \qquad
H^{(m)}(\bar{{}\av}^*) = H^{(m)}({}\av^*)
\quad \text{ if } \quad {}\av \not\in A_{p+1}
\end{equation}
which we prove as Lemma~\ref{lem:H-bar-prop} in \cref{apx:equiv}.
Hence $R_{{}\av}^{(m)} = Q_{{}\av}$ if ${}\av \in A_{p+1}$ and $R_{\bar{{}\av}}^{(m)} = -R_{{}\av}^{(m)}$ if ${}\av\not\in A_{p+1}$.
Lastly, Ref.~\cite{farhi2019quantum} defines
for any $\av\in A$ and $\bv\in A\setminus A_{p+1}$
\begin{equation}
    X_\av := Q_\av \exp\Big(-\frac12 \sum_{\bv\in A_{p+1}} Q_\bv \Phi_{\av\bv}^2\Big) 
    \qquad \text{ and } \qquad
    \Delta_{\bv, \av} := \frac12(\Phi_{\av\bar\bv}^2 - \Phi_{\av\bv}^2).
\end{equation}
Putting everything together, we can write \eqref{eq:R-iter} as
\begin{align}
R^{(m)}_{{}\av} 
&=  Q_{{}\av} \exp \Big[ {-} \frac12\sum_{{}\bv\in A_{p+1}} Q_{{}\bv} \Phi_{{}\av {}\bv}^2 - \frac14 \sum_{{}\bv\not\in A_{p+1}} R_{{}\bv}^{(m-1)} (\Phi_{{}\av {}\bv}^2  - \Phi_{{}\av \bar{{}\bv}}^2)\Big] \nonumber \\
&=  Q_{{}\av} \exp \Big( {-} \frac12\sum_{{}\bv\in A_{p+1}} Q_{{}\bv} \Phi_{{}\av {}\bv}^2 + \frac12 \sum_{{}\bv\not\in A_{p+1}} R_{{}\bv}^{(m-1)} \Delta_{{}\bv, {}\av}\Big) \nonumber \\
&= X_{{}\av} \exp \Big(\frac12 \sum_{{}\bv\not\in A_{p+1}} R_{{}\bv}^{(m-1)} \Delta_{{}\bv, {}\av}\Big).
\label{eq:R-iter-again}
\end{align}

We now want to show that $R^{(p)}_\av$ is a fixed point of the iteration in Eq.~\eqref{eq:R-iter-again}, by showing that $H^{(p)}(\av)$ is a fixed point of the iteration in \eqref{eq:Hm-again}.
Combining Eqs.~\eqref{eq:H_tilde_m_grouped} and \eqref{eq:G_definition}, we can write
\begin{equation}
    H^{(p)}(\av) =
    \exp\Big[{-}
    \frac12\sum_{j,k=-p}^p G_{j,k}^{(p-1)} \Gamma_{j} \Gamma_{k} a_{j} a_{k}
    \Big] 
\end{equation}
As discussed in \cref{sec:infinite_D_iteration} (and proved in Appendix~\ref{apx:proof_placement_G}), after $p-1$ steps of the iteration, all matrix elements $G^{(p-1)}_{j,k}$ reach their final value except when either $j=0$ or $k=0$.
In other words, $G^{(p-1)}_{j,k} = G^{(p)}_{j,k}$ except when $j=0$ or $k=0$.
Noting that $\Gamma_0 = 0$, we have
\begin{equation}
    H^{(p)}(\av) = 
    \exp\Big[{-}
    \frac12\sum_{j,k=-p}^p G_{j,k}^{(p)} \Gamma_{j} \Gamma_{k} a_{j} a_{k}\Big] = \exp \Big[ {-}\frac{1}{2} \sum_\bv f(\bv) H^{(p)}(\bv) \big(\Gammav \cdot (\av \bv) \big)^2 \Big].
    \label{eq:H-p-fixedpoint2}
\end{equation}
where we plugged back in the definition \eqref{eq:G_definition} of $G^{(p)}_{j,k}$.
This means $H^{(p)}(\av)$ is a fixed point of the iteration in \cref{eq:H_tilde_m}, which also implies $R^{(p)}_\av$ is also a fixed point of the iteration \eqref{eq:R-iter-again}:
\begin{align} \label{eq:R-fixedpoint}
R^{(p)}_\av &= X_{{}\av} \exp \Big(\frac12 \sum_{{}\bv\not\in A_{p+1}} R_{\bv}^{(p)} \Delta_{{}\bv, {}\av}\Big).
\end{align}

We can simplify this further.
Note in general for $\bv \not\in A_{p+1}$, we have $R^{(p)}_{\bar\bv} \Delta_{\bar\bv,\av} = R^{(p)}_\bv \Delta_{\bv,\av}$.
Now, let us bipartition $A\setminus A_{p+1} = D\cup D^c$, such that if $\bv\in D$ then $\bar\bv \in D^c$.
Then we can rewrite the above as
\begin{equation}
R^{(p)}_\av = X_\av \exp\Big(\sum_{\bv\in D} R_{\bv}^{(p)} \Delta_{\bv, \av}\Big).
\end{equation}
Moreover, as stated in Section 4 of Ref.~\cite{farhi2019quantum}, for every element of $\bv\in D$, there is a method to assign a unique index $j(\bv) \in \{1,2,\ldots, |D|\}$ such that if $j(\bv) \le j(\bv')$ then $\Delta_{\bv,\bv'} = 0$.
Thus, we have $R_{|D|}^{(p)} = X_{|D|}$. The remaining $R^{(p)}_j$ for  $j=|D|-1,\ldots,2,1$ are then determined via the following relation
\begin{equation}
R^{(p)}_j = X_j \exp\bigg({ \sum_{k=j+1}^{|D|} R_{k}^{(p)} \Delta_{k, j} }\bigg).
\end{equation}
which gives an alternative method to evaluate \eqref{eq:R-iter} to get $R^{(p)}_\av$ for all $\av \in D$.
For the rest, we have $R^{(p)}_{\bar\av} = -R^{(p)}_\av$ for $\bar\av \in D^c$, and $R^{(p)}_\av = Q_\av$ for $\av \in A_{p+1}$.
This is precisely the iteration that yields $W_\av$ described in Section 4 of Ref.~\cite{farhi2019quantum}. Hence
\begin{equation}
    W_\av = R_\av^{(p)} = Q_{{}\av} H^{(p)}({}\av^*) 
    \qquad\text{for any }
    \av \in A.
\end{equation}

Finally, it remains to show $\nu_p(\paramv) = V_p(\paramv)$.
Note we can rewrite the version of $\nu_p(\paramv) = \lim_{D \rightarrow \infty} \nu_p(D,\paramv)$ in Eq.~\eqref{eq:inf_D_proof_with_H_tilde} as 
\begin{align}
\nu_p(\paramv) &= \frac{i}{2} \sum_{\av, \bv \in B} a_0 b_0 f(\av) f(\bv) H^{(p)}(\av) H^{(p)}(\bv) \Gammav \cdot (\av \bv)  \nonumber\\
&= \frac{i}{2} \sum_{\av, \bv \in B} \frac{1}{2}Q_{\lo\av} \frac12 Q_{\lo\bv} H^{(p)}(\lo\av^*) H^{(p)}(\lo\bv^*) \Phi_{\lo\av\lo\bv} 
\end{align}
where we used Eqs.~\eqref{eq:f-and-Q}, \eqref{eq:Phi-and-Gamma} and \eqref{eq:Hm-lo-star}.
Since the summand is independent of $a_0$ and $b_0$, we can sum over $\lo\av,\lo\bv\in A$ instead of $\av,\bv\in B$, killing both $1/2$ factors to get
\begin{align}
\nu_p(\paramv) =
\frac{i}{2} \sum_{\lo\av, \lo\bv \in A }Q_{\lo\av} Q_{\lo\bv} H^{(p)}(\lo\av^*) H^{(p)}(\lo\bv^*) \Phi_{\lo\av\lo\bv} &= \frac{i}{2} \sum_{\lo\av, \lo\bv \in A} \Phi_{\lo\av\lo\bv} W_{\lo\av} W_{\lo\bv} \nonumber \\
&= \frac{i}{2} \sum_{\uv, \vv \in A} \Phi_{\uv\vv} W_{\uv} W_{\vv} = V_p(\paramv)
\end{align}
where in the last line we replaced $\lo\av,\lo\bv$ with dummy variables $\uv,\vv$ since they are summed over. 
The last equality follows from the formula of $V_p(\paramv)$ detailed in Section 4 of Ref.~\cite{farhi2019quantum}.
This proves \cref{thm:SKequiv}.

So we have shown that the performance of the QAOA on any large-girth $D$-regular graph in the $D\to\infty$ limit is equivalent to its ensemble-averaged performance on the SK model in the infinite size limit.

\section{Conjecture that our iteration achieves the Parisi value}
\label{sec:conjecture}

The cut fraction output by the QAOA on MaxCut for large-girth $(D+1)$-regular graphs is
\begin{equation}
    \frac{\braket{\paramv | C_{\text{MC}} | \paramv}}{|E|} = \frac{1}{2} + \frac{\nu_p(D,\paramv)}{\sqrt{D}}.
\end{equation}
We have given an iteration for evaluating $\nu_p(D,\paramv)$ for any depth $p$ and parameters $\paramv$.
Furthermore, in \cref{sec:infinite_D_iteration} we give a compact iteration for 
$\nu_p(\paramv) = \lim_{D \rightarrow \infty} \nu_p(D,\paramv)$.
Using this iteration we can optimize over parameters to get
$\bar\nu_p = \max_{\paramv} \nu_p(\paramv)$.
Note $\bar\nu_p$ cannot be bigger than the Parisi value, $\Pi_*=\lim_{n\to\infty} \EV_J [\max_{\zv} C_J^{\rm SK}(\zv)/n]$.  From our numerics out to $p=17$ we see that $\bar\nu_p$ is headed in that direction.

Now we make the bold conjecture:
\begin{conj*}
Let $\Pi_* = 0.763166...$ be the Parisi value \cite{parisi1979toward, schmidt2008replica}. Then
\begin{equation}
    \lim_{p \rightarrow \infty} \bar\nu_p = \Pi_* .
\end{equation}
\end{conj*}
That is, the iteration in \cref{sec:infinite_D_iteration} is an alternative procedure to compute $\Pi_*$.
To prove this conjecture,
perhaps one can show that the iteration in this paper is equivalent to one of the known procedures for computing $\Pi_*$.
(It may be interesting to note that $\Pi_* = \lim_{k\to\infty} \mathscr{P}_k$, where $\mathscr{P}_k$ is the minimum of the Parisi variational principle over a $k$-step replica symmetry breaking ansatz with $2k+1$ parameters~\cite{panchenko2013sherrington,auffinger2020RSB}. This is not unlike $\bar\nu_p$.)
Or one can find a way to analytically evaluate the $p\to\infty$ limit.

There is an order of limits issue we now address. For any combinatorial optimization problem of fixed size, the QAOA can be shown to give the optimal solution in the $p \to \infty$ limit~\cite{farhi2014quantum}. This may require $p$ to grow exponentially in the system size. But we calculate the performance $\bar\nu_p$ of the QAOA at fixed $p$ in the $D \to \infty$ limit (which means infinite system size). Then we take $p \to \infty$. Our conjecture is about whether, under this new order of limits, the QAOA achieves the optimum as $p\to\infty$.


\clearpage

\section{Generalized iterations for Max-$q$-\XORSAT}
\label{sec:Max-q-XORSAT}
It turns out we can easily generalize our iterations for the QAOA's performance on MaxCut in \cref{sec:QAOA_tree_iterations} to the Max-$q$-{\XORSAT} problem. MaxCut is  a special case of Max-$2$-{\XORSAT}.
Given a $q$-uniform hypergraph $G=(V,E)$ where $E\subseteq V^q$, and given a signed weight $J_{i_1i_2\ldots i_q} \in \{\pm 1\}$ for each edge $(i_1,i_2,\ldots, i_q)\in E$,
Max-$q$-{\XORSAT} is the problem of maximizing the following cost function:
\begin{equation} \label{eq:CXOR}
    C^\XOR_J(\zv) =  \sum_{(i_1,\ldots, i_q)\in E} \frac12 (1 + J_{i_1i_2\ldots i_q} z_{i_1} z_{i_2} \cdots z_{i_q}).
\end{equation}
This cost function can be understood as counting the number of satisfied clauses, where a clause is satisfied if $z_{i_1} z_{i_2}\cdots z_{i_q}=J_{i_1i_2\ldots i_q}$ on the associated edge.
Note the MaxCut cost function in \cref{eq:MaxCut} is a special case of this problem where $q=2$ and all $J_{i_1i_2}=-1$.

We consider this problem on $(D+1)$-regular hypergraphs, where each vertex has degree $D+1$, i.e., it is part of exactly $D+1$ hyperedges. 
(As in \cref{sec:background}, working with $(D+1)$-regular hypergraphs means the subgraphs that the QAOA sees are $D$-ary hypertrees.)
The total number of hyperedges is $|E|=n(D+1)/q$, where $n=|V|$ is the number of vertices.
Due to a result by Sen~\cite{SenHypergraph}, we know that with high probability as $n\to\infty$, the maximum fraction of satisfied clauses for a random $(D+1)$-regular hypergraph for sufficiently large $D$ is 
\begin{equation}
\label{eq:general-parisi}
 \frac{1}{|E|} \max_\zv C^\XOR_J(\zv) =  \frac12 + \Pi_q \sqrt{\frac{q}{2D}} + o(1/\sqrt{D})
\end{equation}
where $\Pi_q$ is the generalized Parisi value that can be determined explicitly.%
\footnote{See Ref.~\cite{SenHypergraph} for how this value can be calculated. Take care to note that the conventions slightly differ, and our $\Pi_q = \mathsf{P}_q/\sqrt{2}$ where $\mathsf{P}_q$ is defined in Section 2.1 of Ref.~\cite{SenHypergraph}.}
In particular, $\Pi_2 = \Pi_* = 0.763166\ldots$.

We want to evaluate how the QAOA performs on the Max-$q$-{\XORSAT} problem for large-girth $(D+1)$-regular hypergraphs.
Here, girth is defined as the minimum length of Berge cycles in the hypergraph~\cite{Berge}.
Similar to the MaxCut problem discussed in \cref{sec:background}, we will see that the QAOA has optimal parameters $\gammav$ that are of order $1/\sqrt{D}$ for these graphs.
For this reason, it will be convenient to prepare the QAOA state $\ket{\paramv}_{J}$ with the following shifted and scaled cost function operator
\begin{equation}
    C_J = \frac{1}{\sqrt{D}} \sum_{(i_1,\ldots, i_q)\in E} J_{i_1i_2\ldots i_q} Z_{i_1} Z_{i_2} \cdots Z_{i_q}.
\end{equation}
For any such hypergraph, we are interested in the fraction of satisfied clauses output by the QAOA at any parameters, for any choices of $J_{i_1 i_2\ldots i_q}$ drawn from $\{+1, -1\}$.
We show the following:
\begin{thm}\label{thm:qXORSAT}
Consider $C^\XOR_J$ on any $(D+1)$-regular $q$-uniform hypergraphs with girth $> 2p+1$.
Let $\ket{\paramv}_{J} $ be the QAOA state generated using $C_{J}$.
Then for any choice of $J$,
\begin{equation} \label{eq:sat-frac}
\frac{1}{|E|} \tensor[_{J}]{\braket{\paramv | C^\XOR_J | \paramv }}{_{J}}  = \frac12 + \nu_{p}^\qsc(D, \paramv) \sqrt{\frac{q}{2D}}
\end{equation}
where $\nu_{p}^\qsc(D, \paramv)$ is independent of $J$ and can be evaluated (on a classical computer) with an iteration using $O(p4^{pq})$ time and $O(4^p)$ memory.
In the infinite $D$ limit, $\lim_{D\to\infty}\nu_{p}^\qsc(D, \paramv)$ can be evaluated with an iteration using $O(p^2 4^p)$ time and $O(p^2)$ memory.
\end{thm}

In \cref{sec:J-ind} that follows, we prove the $J$-independence of $\nu_p^\qsc$ and discuss its implication for a worst-case algorithmic threshold.
We then state and prove the iterations for any finite $D$ and the infinite $D$ limit in Sections~\ref{sec:finite_D_iteration_qxor} through \ref{sec:infinite_D_proof_qxor}.
We also present results of numerical evaluation of the infinite $D$ iteration in \cref{sec:qXOR-numerics}.

\begin{figure}[htb]
\centering
\includegraphics[width=\linewidth]{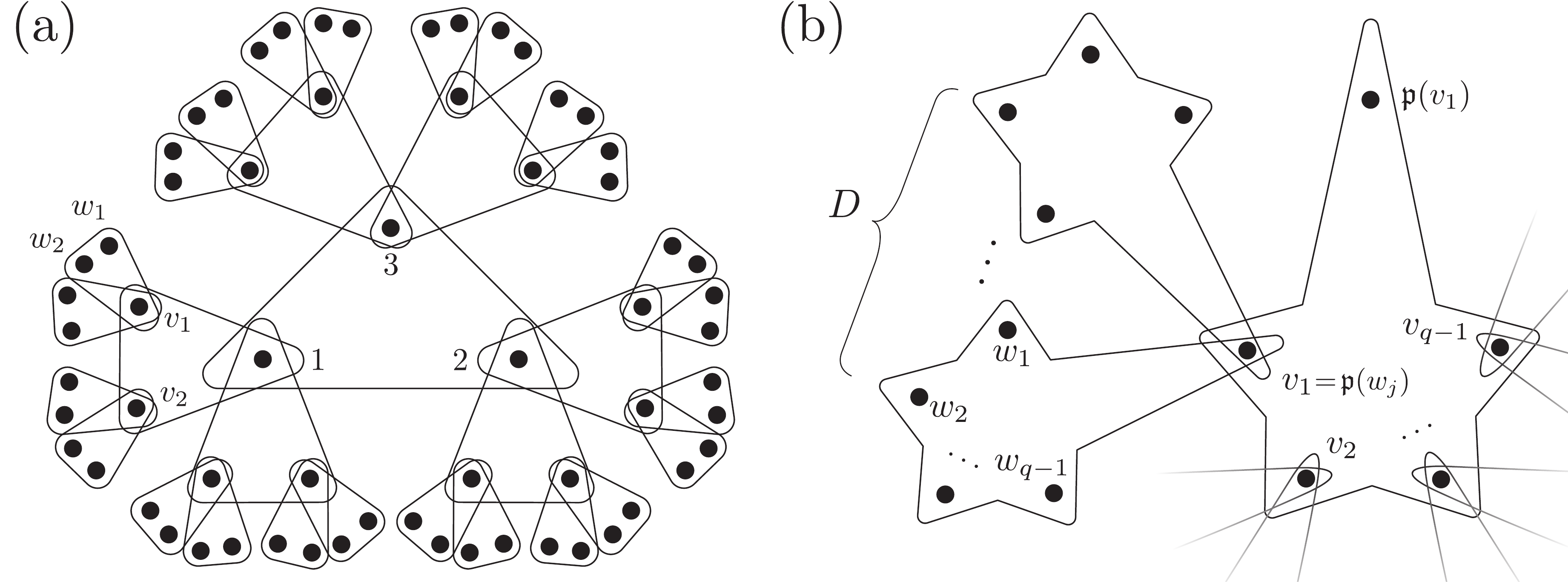}
\vspace{-10pt}
\caption{(a) The hypertree subgraph seen by the QAOA at $p=2$ for the hyperedge $(1,2,\ldots,q)$ on a $(D+1)$-regular $q$-uniform hypergraph with girth $>2p+1$, for $q=3$ and $D=2$.
(b) A partial view near the leaves of the hypertree subgraph for a general $q$ and $D$. The starfish are hyperedges.  Here $w_1,w_2,\ldots,w_{q-1}$ are leaf nodes in the same hyperedge, and we denote their common parent as $v_1=\parent(w_1)=\cdots=\parent(w_{q-1})$.}
\label{fig:qXOR}
\end{figure}

\subsection{$J$-independence of $\nu_p^\qsc$ and implied worst-case limitation}
\label{sec:J-ind}

We start by arguing that the left hand side of \cref{eq:sat-frac} is independent of the choice of $J$'s, so there is no $J$ needed on the right hand side.
When the girth of the hypergraph is larger than $2p+1$, the subgraph seen by the QAOA on any hyperedge is always a $D$-ary $q$-uniform hypertree.
See \cref{fig:qXOR}(a) for an example.
In this figure each triangle is associated with a coupling $J$ that can be either $+1$ or $-1$.
Look at the triangle containing vertices $1,2$ and $3$.
We can absorb the sign of $J_{123}$ into the bit at vertex 1 as follows:
if $J_{123} = -1$ do nothing, whereas if $J_{123} = +1$ flip the sign of the bit at vertex $1$ by redefining $Z_1 \to - Z_1$.
Then $J_{123}Z_1 Z_2 Z_3 \to -Z_1 Z_2 Z_3$ under this transformation.
Now look at the triangle containing bits 1, $v_1$ and $v_2$.
The sign of  $J_{1v_1v_2}$ may have been modified by the last step.
But we can now absorb the sign of $ J_{1v_1v_2}$ 
into the bit at $v_1$ so that $J_{1v_1v_2} Z_1 Z_{v_1} Z_{v_2} \rightarrow - Z_1 Z_{v_1} Z_{v_2}$.
This might affect the sign of  $J_{v_1 w_1 w_2}$ in the triangle containing $v_1$, $w_1$ and $w_2$.
But we can redefine the bit at $w_1$ appropriately so that $J_{v_1 w_1 w_2} Z_{v_1} Z_{w_1} Z_{w_2} \rightarrow -Z_{v_1} Z_{w_1} Z_{w_2}$.
Since there are no cycles in the hypertree, we can move through the whole picture in this way resetting all the couplings $J$ to $-1$.

We have reset all the couplings $J$ to $-1$ in the picture, and we now argue that this makes the quantum expectation \eqref{eq:sat-frac} independent of the $J$'s. At the quantum level we flip the sign of the operator $Z_u$  by conjugating with $X_u$, that is, $X_u Z_u X_u = - Z_u$.
Since the driver $B$ commutes with each $X_u$ and the initial state is an eigenstate of each $X_u$, we can sprinkle $X_u$'s into the left hand side of Eq.~\eqref{eq:sat-frac} and establish the $J$-independence of the expression coming from any particular hyperedge.
Now the cost function \eqref{eq:CXOR} is a sum over the hyperedges of a given hypergraph, but the expected value of each term in the QAOA state is independent of the $J$'s.
So for every $(D+1)$-regular $q$-uniform hypergraph with girth $> 2p+1$ we can write 
\begin{align}
\frac{1}{|E|} \tensor[_{J}]{\braket{\paramv | C^\XOR_J | \paramv }}{_{J}}   = \frac12 - \frac12
{\braket{\paramv |  Z_{1} Z_{2}\ldots  Z_{q} | \paramv }}  
\end{align}
where $(1,2,3,\ldots,q)$ is any hyperedge, and the state $\ket{\paramv}$ without the $J$ label has all the couplings set to $-1$.
In the sections that follow, we provide an iterative formula for
\begin{align} \label{eq:nu-p-q}
\nu_{p}^\qsc(D, \paramv) = -\sqrt{\frac{D}{2q}}{\braket{\paramv |  Z_{1} Z_{2}\ldots  Z_{q} | \paramv }} 
\end{align}
which gives the QAOA performance for Max-$q$-{\XORSAT}  at any parameters on any $(D+1)$-regular $q$-uniform hypergraph with girth $> 2p+1$, regardless of the choices of $J$. This generalizes Eq.~\eqref{eq:implicit_def_eta}.

A corollary to this $J$-independence is that the QAOA at low depth fails to find the optimal assignment in the worst case.
To see this, let us go back to the $q=2$ case where we studied MaxCut on a large-girth regular graph which has all of the couplings $J=-1$.
At optimal parameters, the fraction of satisfied clauses is $1/2 + \bar\nu_p/\sqrt{D}$ in the large $D$ limit, where $\bar{\nu}_p \le \Pi_*$.
Consider the corresponding instance where all the couplings on the same graph are set to $J=+1$, which makes the instance fully satisfiable. In that case, the best possible fraction of satisfied clauses is 1. However, the fraction output by the QAOA at optimal parameters is the same as in the $J=-1$ case, that is, at most $1/2 + \Pi_*/\sqrt{D}$, which is only a bit more than $1/2$ in the large $D$ limit.

Here we have an example of the QAOA failing to reach the optimum in the worst case because it does not ``see'' the whole graph. (Unlike previous results of the similar flavor in Refs.~\cite{BKKT, FGG20}, we do not need the graph to be bipartite to bound the worst-case approximation ratio.)
Regardless of the signs of the couplings,  
the low-depth QAOA sees a tree subgraph surrounding each edge.
On the tree subgraph the signs of the couplings are irrelevant so the QAOA does not distinguish between instances where the cost function favors disagreement and instances where agreement is favored.
Without seeing cycles the QAOA cannot do better than what it can achieve in the most frustrated case, and this yields an upper bound on the worst-case approximation ratio.

\subsection{An iteration for any finite $D$}\label{sec:finite_D_iteration_qxor}

Here we give an iteration to evaluate $\nu_{p}^\qsc(D, \paramv)$ for any input parameters, $q$ and $D$.
We prove that this iteration is correct in \cref{sec:proof_finite_D_qxor}.
We use the same convention as in \cref{sec:QAOA_tree_iterations}, where we defined a $(2p+1)$-component vector $\Gammav$ and a function $f(\av)$ that takes any $(2p+1)$-bit string $\av$ as input.
Similar to the iteration in \cref{sec:finite_D_iteration}, we start with $H_D^{(0)}(\av) = 1$, and for  $1 \leq m \leq p$, let
\begin{align} \label{eq:H_m_def_general}
    H_D^{(m)}(\av) &= \bigg(\sum_{\bv^1, \ldots, \bv^{q-1}} 
     \cos{\Big[ {\textstyle \frac{1}{\sqrt{D}}}\Gammav \cdot (\av \bv^1 \bv^2 \cdots \bv^{q-1}) \Big]}
     \prod_{i=1}^{q-1} \big[f(\bv^i) H_D^{(m-1)}(\bv^i)\big]
      \bigg)^D.
\end{align}
By iteratively evaluating \cref{eq:H_m_def_general} for $m=1,2,\ldots,p$, we arrive at $H_D^{(p)}(\av)$ which is used to compute
\begin{align} \label{eq:nu-p-q-D-formula}
    \nu_p^\qsc(D,\paramv) = i \sqrt{\frac{D}{2q}} \sum_{\av^1,\ldots, \av^q} 
    \sin{\Big[ {\textstyle \frac{1}{\sqrt{D}}} \Gammav \cdot (\av^1\av^2\cdots \av^q) \Big]}
    \prod_{i=1}^q \big[ a^i_0 f(\av^i) H_D^{(p)}(\av^i)\big].
\end{align}

We note that this iteration on $H_D^{(m)}$ has $p$ steps, each involving a sum with $2^{(2p+1)(q-1)}$ terms for each of the $2^{2p+1}$ entries of $H_D^{(m)}(\av)$.
The final step has a sum with $O(4^{pq})$ terms. Overall, this iteration has a time complexity of $O(p\, 4^{pq})$ and a memory complexity of $O(4^p)$ for storing the entries of $H^{(m)}_D$.

\subsection{An iteration for $D \rightarrow \infty$}\label{sec:infinite_D_iteration_qxor}
In the infinite $D$ limit, we get a more compact iteration which we prove is correct in \cref{sec:infinite_D_proof_qxor}. Similar to \cref{sec:infinite_D_iteration}, we define matrices $G^{(m)} \in \mathbb{C}^{(2p+1) \times (2p+1)}$, 
for $0 \leq m \leq p$ as follows. For $j,k \in \{1,\dots,p, 0,-p,\dots,-1\}$, let $G_{j,k}^{(0)} = \sum_{\av} f(\av) a_j a_k$
and
\begin{equation} \label{eq:inf_D_iter_G_m_general}
    G_{j,k}^{(m)} = \sum_{\av} f(\av) a_j a_k \exp
    \Big[{-}\frac{1}{2} \sum_{j',k'=-p}^p \big(G_{j',k'}^{(m-1)}\big)^{q-1} \Gamma_{j'} \Gamma_{k'} a_{j'} a_{k'} 
    \Big]
    \quad
    \text{for } 1 \leq m \leq p.
\end{equation}
Starting at $m=0$ and going up by $p$ steps we arrive at $G^{(p)}$ which is used to compute
\begin{equation} \label{eq:infinite_D_iter_result_general}
     \nu_p^\qsc(\paramv) := \lim_{D \rightarrow \infty} \nu_p^\qsc(D,\paramv) 
     = \frac{i}{\sqrt{2q}} \sum_{j = -p}^p \Gamma_j ( G^{(p)}_{0, j} )^q.
\end{equation}

Note the only difference between Max-$q$-{\XORSAT} 
and MaxCut, where $q=2$, can be seen by comparing Eqs.~\eqref{eq:inf_D_iter_G_m} and \eqref{eq:infinite_D_iter_result} in \cref{sec:infinite_D_iteration} to Eqs.~\eqref{eq:inf_D_iter_G_m_general} and \eqref{eq:infinite_D_iter_result_general} in the current iteration, where we are raising the matrix element of $G$ to some $q$-dependent power.
This iteration takes at most $O(p^2 4^p)$ time and $O(p^2)$ memory to evaluate using the same method as described in \cref{apx:proof_placement_G}, regardless of $q$.
This is polynomially faster than the finite $D$ case with exponentially better memory usage.

\subsection{Proof of the finite $D$ iteration}\label{sec:proof_finite_D_qxor}
We now prove that the iteration in \cref{sec:finite_D_iteration_qxor} for finite degree $D$ is correct.
The proof is essentially the same as in \cref{sec:iter_proofs_finite}, and we will only focus on the differences in what follows.
The goal is to evaluate the expectation \eqref{eq:nu-p-q} on a single hyperedge $(1,2,\ldots,q)$ by restricting to the hypertree subgraph seen by the QAOA.
As an example, we show such a subgraph in \cref{fig:qXOR}(a).
After inserting complete sets and reorganizing terms as we have done from Eq.~\eqref{eq:ZLZR-p=2} to Eq.~\eqref{eq:ZLZR-deriv-2}, we arrive at
\begin{align}
\label{eq:EJ-Z1-q}
&
\braket{\paramv | Z_{1} Z_{2} \cdots Z_q | \paramv } =
\sum_{\{\zv_u \}}
z^\pp{0}_1 z^\pp{0}_2 \cdots z^\pp{0}_q
 \exp\Big[ {\textstyle -\frac{i}{\sqrt{D}}}  \sum_{(i_1,\ldots,i_q)\in E} \Gammav  \cdot (\zv_{i_1} \zv_{i_2} \cdots \zv_{i_q}) \Big]
 \prod_{v=1}^n f(\zv_v)
\end{align}
which is analogous to Eq.~\eqref{eq:ZLZR-deriv-2}.
To evaluate this, we again need to sum over all the bit configurations $\zv_u=(z_u^\pp{1},\ldots,z_u^\pp{p},z_u^\pp{0},z_u^\pp{-p},\ldots,z_u^\pp{-1})$ of every node $u$ in the hypertree subgraph.
We will do this with the same method as in \cref{sec:iter_proofs_finite}, where we first sum over all the leaf nodes, then their parents, and their parents' parents, and so on.

We start by summing over a set of leaf nodes $w_1, w_2,\ldots, w_{q-1}$ which are in the same hyperedge as their parent. Let their common parent be $\parent(w_1)$.
See \cref{fig:qXOR}(b) for a visualization.
Summing over the bit configurations $\zv_{w_1}, \zv_{w_2},\ldots, \zv_{w_{q-1}}$ yields
\begin{align}
&\sum_{\zv_{w_1},\ldots, \zv_{w_{q-1}}} 
\exp\Big[ {\textstyle -\frac{i}{\sqrt{D}}}  \Gammav \cdot (\zv_{\parent(w_1)}\zv_{w_1}\zv_{w_2} \cdots \zv_{w_{q-1}} )
 	\Big]
	\prod_{j=1}^{q-1}  f(\zv_{w_j}) \,.
\end{align}
Note this is a function of the parent node's bit configuration $z_{\parent(w_1)}$.
Since $\parent(w_1)$ is involved in exactly $D$ branching hyperedges, each of which contains $q-1$ distinct children, we get the following contribution
\begin{align}
H_D^{(1)}(\zv_{\parent(w_1)}) :=
\bigg(\sum_{\zv_{w_1},\ldots, \zv_{w_{q-1}} }  
 \exp\Big[ {\textstyle -\frac{i}{\sqrt{D}}} \Gammav \cdot (\zv_{\parent(w_1)}\zv_{w_1}\zv_{w_2} \cdots \zv_{w_{q-1}} )
 				\Big]  
	\prod_{j=1}^{q-1}  f(\zv_{w_j}) 
	\bigg)^D
\end{align}
after summing over all the leaf nodes. This applies to every parent node of any of the leaves.
And similar to what we have done to get \cref{eq:HDm-final} in \cref{sec:iter_proofs_finite}, we can use the fact that $f(-\zv)=f(\zv)$ to take $\zv_{w_1}\to -\zv_{w_1}$ in the above summand and combine it with its original form to get
\begin{align} \label{eq:HD1-contribution}
H_D^{(1)}(\zv_{\parent(w_1)}) =
\bigg(\sum_{\zv_{w_1},\ldots, \zv_{w_{q-1}} }  
 \cos\Big[ {\textstyle \frac{1}{\sqrt{D}}} \Gammav \cdot (\zv_{\parent(w_1)}\zv_{w_1}\zv_{w_2} \cdots \zv_{w_{q-1}} )
 				\Big]  
	\prod_{j=1}^{q-1}  f(\zv_{w_j}) 
	\bigg)^D .
\end{align}

Next, we repeat this argument for all the parent nodes like $\parent(w_1)$.
Let $v_1 = \parent(w_1)$, and let $v_2, v_3, \ldots, v_{q-1}$ be the other nodes in the same hyperedge as $v_1$ [see \cref{fig:qXOR}(b)].
We also denote $\parent(v_1)$ as their shared parent node.
Including the contribution $H_D^{(1)}(\zv_{v_j})$ coming from the above sum over the leaves, we sum over all $\zv_{v_j}$'s to get the following function on $\zv_{\parent(v_1)}$
\begin{align}
H_D^{(2)}(\zv_{\parent(v_1)}) &:=
\bigg(\sum_{\zv_{v_1},\ldots, \zv_{v_{q-1}} }  
 \exp\Big[ {\textstyle -\frac{i}{\sqrt{D}}} \Gammav \cdot (\zv_{\parent(v_1)}\zv_{v_1}\zv_{v_2} \cdots \zv_{v_{q-1}} )
 	 \Big]
\prod_{j=1}^{q-1}  \big[f(\zv_{v_j}) H_D^{(1)}(\zv_{v_j})\big]
\bigg)^D \nonumber \\
&=\bigg(\sum_{\zv_{v_1},\ldots, \zv_{v_{q-1}} }  
 \cos\Big[ {\textstyle \frac{1}{\sqrt{D}}} \Gammav \cdot (\zv_{\parent(v_1)}\zv_{v_1}\zv_{v_2} \cdots \zv_{v_{q-1}} )
 	 \Big]
\prod_{j=1}^{q-1}  \big[f(\zv_{v_j}) H_D^{(1)}(\zv_{v_j})\big]
\bigg)^D
\end{align}
where the last equality follows from the fact that $f(-\zv)=f(\zv)$ and $H_D^{(1)}(-\zv) = H_D^{(1)}(\zv)$.

It is easy to see that continuing in this fashion we can iteratively sum over all but the root nodes in $p$ steps, corresponding to the $p$ levels in the hypertree subgraph.
At each step $m$, we obtain $H_D^{(m)}(\zv)$ which is a function of the bit configuration $\zv$ of any node that is $m$ levels away from the leaves.
For consistency of notation, we let $H_D^{(0)}(\zv)=1$.
At the end of the iteration, we reach the top-level root nodes inside the central hyperedge $(1,2,\ldots, q)$, and we can evaluate Eq.~\eqref{eq:EJ-Z1-q} as
\begin{align}
\braket{\paramv| Z_1 Z_2\cdots Z_q|\paramv }
&=\sum_{\zv_1,\ldots,\zv_q}
z^\pp{0}_1 z^\pp{0}_2 \cdots z^\pp{0}_q
 \exp\Big[ {\textstyle -\frac{i}{\sqrt{D}}}  \Gammav  \cdot (\zv_1 \zv_2 \cdots \zv_q) \Big]
 \prod_{j=1}^q \big[f(\zv_j) H_D^{(p)}(\zv_j)\big] \nonumber \\
&
=-i\sum_{\zv_1,\ldots,\zv_q}
 \sin\Big[ {\textstyle \frac{1}{\sqrt{D}}}  \Gammav  \cdot (\zv_1 \zv_2 \cdots \zv_q) \Big]
 \prod_{j=1}^q \big[z^\pp{0}_j f(\zv_j) H_D^{(p)}(\zv_j)\big].
\end{align}
Then plugging this back into Eq.~\eqref{eq:nu-p-q} gives the iterative formula for $\nu_p^\qsc(D,\paramv)$ in Eq.~\eqref{eq:nu-p-q-D-formula}.

\subsection{Proof of the $D \rightarrow \infty$ iteration}
\label{sec:infinite_D_proof_qxor}

To get the infinite $D$ iteration in \cref{sec:infinite_D_iteration_qxor}, we follow the same argument as in \cref{sec:iter_proofs_infinite}.
This is done by first defining the following limiting functions for $0\le m \le p$
\begin{equation}
H^{(m)}(\av) := \lim_{D\to\infty} H_D^{(m)}(\av).
\end{equation}
These functions can be shown to satisfy the following recursion relation
for $1\le m \le p$
\begin{align}
H^{(m)}(\av) &= \lim_{D\to\infty}\bigg(\sum_{\bv^1, \ldots, \bv^{q-1}} 
     \cos{\Big[ {\textstyle \frac{1}{\sqrt{D}}}\Gammav \cdot (\av \bv^1 \bv^2 \cdots \bv^{q-1}) \Big]}
     \prod_{i=1}^{q-1} \big[f(\bv^i) H_D^{(m-1)}(\bv^i)\big]
      \bigg)^D \nonumber \\
	&=
	\exp\bigg[
	{-}\frac12\sum_{\bv^1, \ldots, \bv^{q-1}} 
	\big(\Gammav \cdot (\av \bv^1 \bv^2 \cdots \bv^{q-1})\big)^2
	\prod_{i=1}^{q-1} \big[f(\bv^i) H^{(m-1)}(\bv^i)\big]
	\bigg]
\end{align}
where the second line is obtained by performing Taylor expansion of the cosine, and using the fact that $\sum_\bv f(\bv) H_D^{(m-1)}(\bv) =1$ (which is analogous to Lemma~\ref{lem:fH-sum} in Appendix~\ref{apx:H-properties} for general $q$).
Expanding the dot product, we get
\begin{equation}
H^{(m)}(\av) =
	\exp\bigg[
	{-}\frac12
	\sum_{j,k=-p}^p \Gamma_j \Gamma_k a_j a_k 
	\prod_{i=1}^{q-1} \Big(\sum_{\bv^i} f(\bv^i) H^{(m-1)}(\bv^i) b^i_j b^i_k\Big)
	\bigg].
	\label{eq:Hm-iteration-general}
\end{equation}
By defining
\begin{equation}
G_{j,k}^{(m)} := \sum_\av f(\av) H^{(m)}(\av) a_j a_k,
\end{equation}
we can recast the iteration \eqref{eq:Hm-iteration-general} as
\begin{equation}
G_{j,k}^{(m)} = \sum_\av f(\av) a_j a_k
\exp\bigg[
	{-}\frac12
	\sum_{j',k'=-p}^p \big(G_{j',k'}^{(m-1)}\big)^{q-1}
	\Gamma_{j'} \Gamma_{k'} a_{j'} a_{k'}
	\bigg]
\end{equation}
for $1\le m \le p$. For $m=0$, we have
\begin{equation}
G_{j,k}^{(0)} = \sum_\av f(\av) a_j a_k,
\end{equation}
since $H^{(0)}(\av) = 1$.
Finally, we can write $\nu_p^\qsc$ in Eq.~\eqref{eq:nu-p-q-D-formula} in the $D\to\infty$ limit as
\begin{align}
   \lim_{D\to\infty} \nu_p^\qsc(D,\paramv) &= \lim_{D\to\infty} i \sqrt{\frac{D}{2q}} \sum_{\av^1,\ldots, \av^q} 
    \sin{\Big[ {\textstyle \frac{1}{\sqrt{D}}} \Gammav \cdot (\av^1\av^2\cdots \av^q) \Big]}
    \prod_{i=1}^q \big[ a^i_0 f(\av^i) H_D^{(p)}(\av^i)\big] \nonumber \\
    &= \frac{i}{\sqrt{2q}} \sum_{j=-p}^p \Gamma_j 
        \prod_{i=1}^q \Big(
		        \sum_{\av^i} a^i_0 a^i_j f(\av^i) H^{(p)}(\av^i)
        	\Big)
    \nonumber \\
    &= \frac{i}{\sqrt{2q}} \sum_{j=-p}^p \Gamma_j \big(G_{0,j}^{(p)}\big)^{q}
\end{align}
which gives the infinite $D$ iterative formula in Eq.~\eqref{eq:infinite_D_iter_result_general} as desired.

\begin{figure}[ht]
    \centering
    \includegraphics[width=\linewidth]{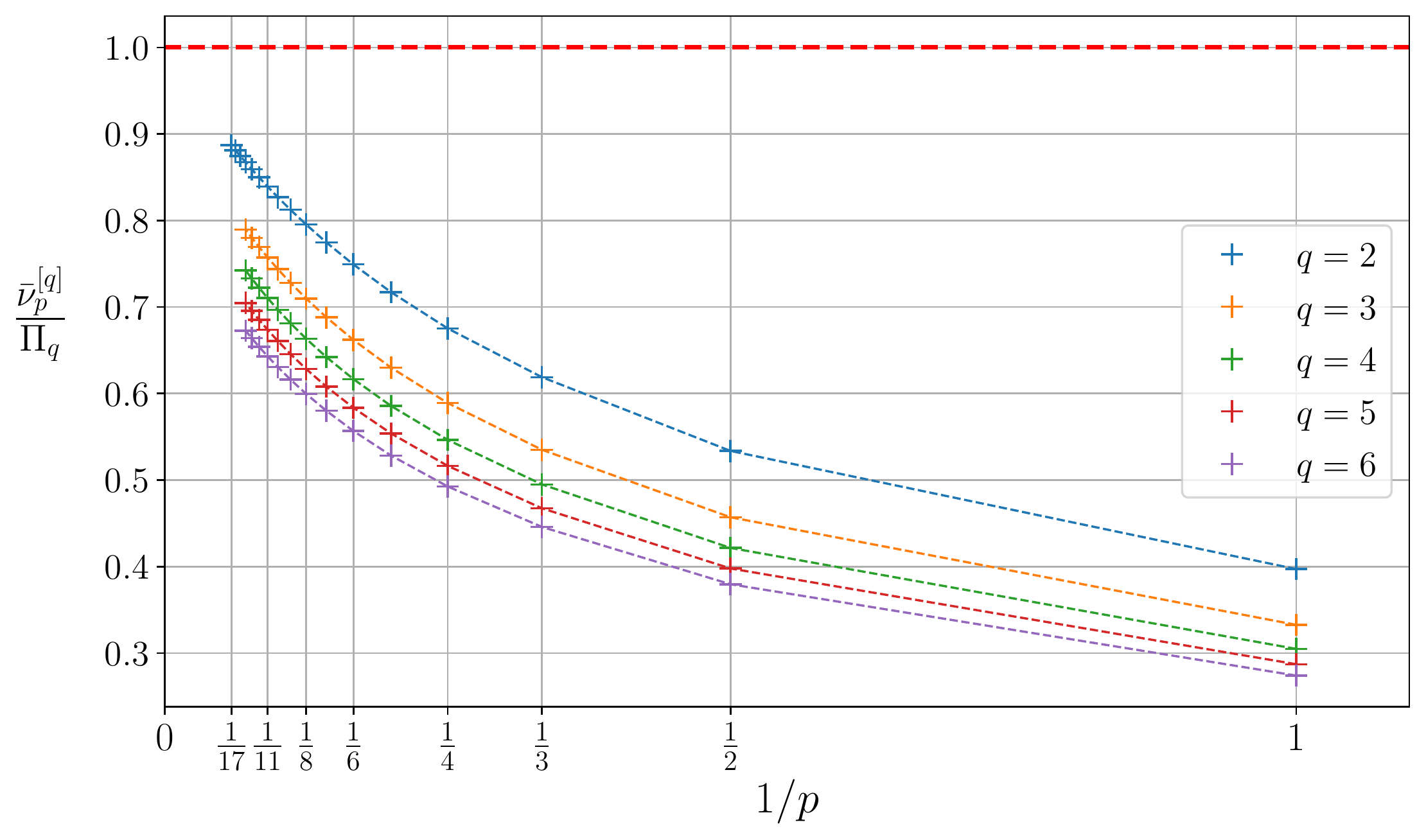}
    \vspace{-10pt}
    \caption{
    Optimal values $\bar{\nu}_p^\qsc$ normalized by their corresponding Parisi values $\Pi_q$ as a function of $1 / p$ for $q = 2, 3, 4, 5, 6$.
    The Parisi values are taken from Ref.~\cite{marwaha2021bounds}.
    Similar to Fig.~\ref{fig:nu_bar}, we made the somewhat arbitrary choice of plotting the data against $1 / p$ to see the large $p$ region in a compact plot.
    Dashed lines in between data points are intended to guide the eye.
    }
    \label{fig:nu_bar_q}
\end{figure}

\vspace{5pt}
\subsection{Numerical evaluation of the Max-$q$-{\XORSAT} performance at infinite $D$}
\label{sec:qXOR-numerics}

We have taken our iteration for infinite $D$ and numerically optimized $\nu_p^\qsc(\paramv)$ to find
\begin{align}
\bar{\nu}_p^\qsc = \max_{\paramv} \nu_p^\qsc(\paramv).
\end{align} 
up to $p=14$ for $3\le q \le 6$.
Combining with the data we have for $q=2$ in \cref{table:optimal_values}, we plot the
results in \cref{fig:nu_bar_q}.
For ease of comparison across different values of $q$ we have normalized $\bar{\nu}_p^\qsc$ by its corresponding Parisi value $\Pi_q$.
See Appendix~\ref{apx:optimal_angles_q} for a plot of the optimal $\gammav$ and $\betav$ we found at $p=14$.
Numerical values for $\bar{\nu}_p^\qsc$ and optimal $\gammav$ and $\betav$ for all $1 \le p \le 14$ can be found in Ref.~\cite{data}.

In some cases, there are thresholds on how well the QAOA at low depths can do.
It is known that for problems that exhibit the overlap gap property, the locality property of the QAOA prevents it from getting close to the optimum at low depths where it does not see the whole graph~\cite{FGG20typical, chou2021limitations}.
Specifically, using an overlap gap property in the Max-$q$-XORSAT problem on random Erd\H{o}s-R\'enyi hypergraphs with constant average degree and even $q\ge 4$, Ref.~\cite{chou2021limitations} showed that the QAOA (or any local algorithm) has limited performance when the depth $p$ is less than $\epsilon \log (n) $, where $n$ is the graph size and $\epsilon$ is a constant that depends on the degree and $q$.
Assuming the overlap gap property also holds when the hypergraphs are regular, one can use similar arguments to show that the QAOA's performance as measured by $\bar{\nu}_p^\qsc/\Pi_q$ does not converge to $1$ as $p\to\infty$ when $q \ge 4$ and is even.
This is because our large-girth assumption implies the graph has at least $D^p$ vertices, so $p$ is always less than $\epsilon \log n$ in this limit.

\section{Discussion}\label{sec:discussion}
In this paper, we have introduced new techniques for evaluating the performance of a quantum algorithm at high qubit number and at high depth.
In particular we do this by finding a compact iteration for the QAOA's performance on MaxCut on instances with locally tree-like neighborhoods.
On random large-girth $D$-regular graphs, the QAOA at $p=11$ and higher has the highest approximation ratio of any assumption-free algorithm.
We have given performance guarantees for the QAOA, but it is necessary to run a quantum computer to produce a string with the calculated performance.

We have also shown that for any depth $p$ and for any parameters,  $\gammav$ and $\betav$, the performance of the QAOA on large-girth $D$-regular graphs, as $D\to\infty$, matches the typical performance of the QAOA on the Sherrington-Kirkpatrick model at infinite size.
We find it remarkable that the ensemble averaging done in the SK model can be replaced by analyzing a single tree subgraph.
For both of these models the best conceivable performance is upperbounded by the Parisi constant, $\Pi_*$.
There are optimal parameters at each $p$, and we speculate that as $p\to\infty$ these optimal parameters give QAOA performance that matches the Parisi constant for both models.

Moreover, in \cref{sec:Max-q-XORSAT}, we have generalized our iteration for MaxCut on large-girth regular graphs to evaluate the QAOA's performance on Max-$q$-{\XORSAT} problems for large-girth regular hypergraphs. 
We have shown that, at fixed parameters, the QAOA gives the same value of the objective function regardless of the signs of the couplings on these hypergraphs.
This implies a worst-case algorithmic threshold at low depth for fully satisfiable instances.
Building on our work, Ref.~\cite{BGMZ22} recently generalized the equivalence between MaxCut and the SK model to between Max-$q$-XORSAT and the fully connected $q$-spin model.

There are a number of ideas to explore coming out of this work. Can we find a more efficient iterative formula for the QAOA's performance than the one in \cref{sec:infinite_D_iteration}?
If so, we can better probe the large-$p$ behavior of the QAOA.
Can the iteration in \cref{sec:infinite_D_iteration} be recast in the $p\to\infty$ limit in terms of continuous functions corresponding to $\paramv$?
This might be a way to verify, or falsify, the conjecture in \cref{sec:conjecture}.

Can one find other problems at high qubit number and high depth where the performance of the QAOA can be established using techniques similar to the ones introduced in this paper?

\section*{Acknowledgements}

The authors thank Sam Gutmann for being there and Matthew P.~Harrigan for a careful read of the manuscript. This material is based upon work supported by the National Science Foundation Graduate Research Fellowship Program under Grant No. DGE-1746045. Any opinions, findings, and conclusions or recommendations expressed in this material are those of the author(s) and do not necessarily reflect the views of the National Science Foundation.

\bibliographystyle{ieeetr_kunal}
\bibliography{refs}

\begin{thebibliography}{10}
\providecommand{\url}[1]{#1}
\csname url@samestyle\endcsname
\providecommand{\newblock}{\relax}
\providecommand{\bibinfo}[2]{#2}
\providecommand{\BIBentrySTDinterwordspacing}{\spaceskip=0pt\relax}
\providecommand{\BIBentryALTinterwordstretchfactor}{4}
\providecommand{\BIBentryALTinterwordspacing}{\spaceskip=\fontdimen2\font plus
\BIBentryALTinterwordstretchfactor\fontdimen3\font minus
  \fontdimen4\font\relax}
\providecommand{\BIBforeignlanguage}[2]{{%
\expandafter\ifx\csname l@#1\endcsname\relax
\typeout{** WARNING: IEEEtran.bst: No hyphenation pattern has been}%
\typeout{** loaded for the language `#1'. Using the pattern for}%
\typeout{** the default language instead.}%
\else
\language=\csname l@#1\endcsname
\fi
#2}}
\providecommand{\BIBdecl}{\relax}
\BIBdecl

\bibitem{trevisan2000gadgets}
\BIBentryALTinterwordspacing
L.~Trevisan, G.~B. Sorkin, M.~Sudan, and D.~P. Williamson, ``{Gadgets,
  approximation, and linear programming},'' \emph{SIAM Journal on Computing},
  vol.~29, no.~6, pp. 2074--2097, 2000.
  \url{https://doi.org/10.1137/S0097539797328847}
\BIBentrySTDinterwordspacing

\bibitem{MontanariSen2016}
\BIBentryALTinterwordspacing
A.~Montanari and S.~Sen, ``Semidefinite programs on sparse random graphs and
  their application to community detection,'' in \emph{Proceedings of the 48th
  Annual ACM Symposium on Theory of Computing}, ser. STOC '16, 2016, pp.
  814--827. \href{https://arxiv.org/abs/1504.05910}{arXiv:1504.05910}
\BIBentrySTDinterwordspacing

\bibitem{lyons2017factors}
\BIBentryALTinterwordspacing
R.~Lyons, ``Factors of iid on trees,'' \emph{Combinatorics, Probability and
  Computing}, vol.~26, no.~2, pp. 285--300, 2017.
  \href{https://arxiv.org/abs/1401.4197}{arXiv:1401.4197}
\BIBentrySTDinterwordspacing

\bibitem{barak2021classical}
\BIBentryALTinterwordspacing
B.~Barak and K.~Marwaha, ``{Classical Algorithms and Quantum Limitations for
  Maximum Cut on High-Girth Graphs},'' in \emph{Proceedings of the 13th
  Innovations in Theoretical Computer Science Conference}, ser. ITCS '22, vol.
  215, 2022, pp. 14:1--14:21.
  \href{https://arxiv.org/abs/2106.05900}{arXiv:2106.05900}
\BIBentrySTDinterwordspacing

\bibitem{thompson2021explicit}
\BIBentryALTinterwordspacing
J.~K. Thompson, O.~Parekh, and K.~Marwaha, ``{An explicit vector algorithm for
  high-girth MaxCut},'' \emph{Symposium on Simplicity in Algorithms (SOSA)},
  pp. 238--246, 2022. \href{https://arxiv.org/abs/2108.12477}{arXiv:2108.12477}
\BIBentrySTDinterwordspacing

\bibitem{farhi2014quantum}
\BIBentryALTinterwordspacing
E.~Farhi, J.~Goldstone, and S.~Gutmann, ``{A Quantum Approximate Optimization
  Algorithm},'' \emph{arXiv preprint}, 2014.
  \href{https://arxiv.org/abs/1411.4028}{arXiv:1411.4028}
\BIBentrySTDinterwordspacing

\bibitem{harrigan2021quantum}
\BIBentryALTinterwordspacing
M.~P. Harrigan, K.~J. Sung, M.~Neeley, K.~J. Satzinger, F.~Arute, K.~Arya,
  J.~Atalaya, J.~C. Bardin, R.~Barends, S.~Boixo \emph{et~al.}, ``{Quantum
  approximate optimization of non-planar graph problems on a planar
  superconducting processor},'' \emph{Nature Physics}, vol.~17, no.~3, pp.
  332--336, 2021. \href{https://arxiv.org/abs/2004.04197}{arXiv:2004.04197}
\BIBentrySTDinterwordspacing

\bibitem{wang2018quantum}
\BIBentryALTinterwordspacing
Z.~Wang, S.~Hadfield, Z.~Jiang, and E.~G. Rieffel, ``{Quantum approximate
  optimization algorithm for MaxCut: A fermionic view},'' \emph{Phys. Rev. A},
  vol.~97, no.~2, p. 022304, 2018.
  \href{https://arxiv.org/abs/1706.02998}{arXiv:1706.02998}
\BIBentrySTDinterwordspacing

\bibitem{wurtz2020bounds}
\BIBentryALTinterwordspacing
J.~Wurtz and P.~Love, ``{MaxCut quantum approximate optimization algorithm
  performance guarantees for $p>1$},'' \emph{Phys. Rev. A}, vol. 103, p.
  042612, 2021. \href{https://arxiv.org/abs/2010.11209}{arXiv:2010.11209}
\BIBentrySTDinterwordspacing

\bibitem{marwaha2021local}
\BIBentryALTinterwordspacing
K.~Marwaha, ``{Local classical MAX-CUT algorithm outperforms $ p= 2$ QAOA on
  high-girth regular graphs},'' \emph{Quantum}, vol.~5, p. 437, 2021.
  \href{https://arxiv.org/abs/2101.05513}{arXiv:2101.05513}
\BIBentrySTDinterwordspacing

\bibitem{alaoui2021local}
\BIBentryALTinterwordspacing
A.~E. Alaoui, A.~Montanari, and M.~Sellke, ``{Local algorithms for Maximum Cut
  and Minimum Bisection on locally treelike regular graphs of large degree},''
  \emph{arXiv preprint}, 2021.
  \href{https://arxiv.org/abs/2111.06813}{arXiv:2111.06813}
\BIBentrySTDinterwordspacing

\bibitem{gamarnik2021overlap}
\BIBentryALTinterwordspacing
D.~Gamarnik, ``{The overlap gap property: A topological barrier to optimizing
  over random structures},'' \emph{Proceedings of the National Academy of
  Sciences}, vol. 118, no.~41, 2021.
  \href{https://arxiv.org/abs/2109.14409}{arXiv:2109.14409}
\BIBentrySTDinterwordspacing

\bibitem{farhi2019quantum}
\BIBentryALTinterwordspacing
E.~Farhi, J.~Goldstone, S.~Gutmann, and L.~Zhou, ``{The Quantum Approximate
  Optimization Algorithm and the Sherrington-Kirkpatrick Model at Infinite
  Size},'' \emph{Quantum}, vol.~6, p. 759, 2022.
  \href{https://arxiv.org/abs/1910.08187v4}{arXiv:1910.08187v4}
\BIBentrySTDinterwordspacing

\bibitem{boulebnane2021predicting}
\BIBentryALTinterwordspacing
S.~Boulebnane and A.~Montanaro, ``{Predicting parameters for the Quantum
  Approximate Optimization Algorithm for MAX-CUT from the infinite-size
  limit},'' \emph{arXiv preprint}, 2021.
  \href{https://arxiv.org/abs/2110.10685}{arXiv:2110.10685}
\BIBentrySTDinterwordspacing

\bibitem{parisi1979toward}
\BIBentryALTinterwordspacing
G.~Parisi, ``{Toward a mean field theory for spin glasses},'' \emph{Physics
  Letters A}, vol.~73, no.~3, pp. 203--205, 1979.
  \url{https://doi.org/10.1016/0375-9601(79)90708-4}
\BIBentrySTDinterwordspacing

\bibitem{dembo2017extremal}
\BIBentryALTinterwordspacing
A.~Dembo, A.~Montanari, and S.~Sen, ``{Extremal cuts of sparse random
  graphs},'' \emph{The Annals of Probability}, vol.~45, no.~2, pp. 1190--1217,
  2017. \href{https://arxiv.org/abs/1503.03923}{arXiv:1503.03923}
\BIBentrySTDinterwordspacing

\bibitem{marwaha2021bounds}
\BIBentryALTinterwordspacing
K.~Marwaha and S.~Hadfield, ``{Bounds on approximating Max $k$XOR with quantum
  and classical local algorithms},'' \emph{arXiv preprint}, 2021.
  \href{https://arxiv.org/abs/2109.10833}{arXiv:2109.10833}
\BIBentrySTDinterwordspacing

\bibitem{ZhouQAOA}
\BIBentryALTinterwordspacing
L.~Zhou, S.-T. Wang, S.~Choi, H.~Pichler, and M.~D. Lukin, ``Quantum
  approximate optimization algorithm: Performance, mechanism, and
  implementation on near-term devices,'' \emph{Phys. Rev. X}, vol.~10, p.
  021067, 2020. \href{https://arxiv.org/abs/1812.01041}{arXiv:1812.01041}
\BIBentrySTDinterwordspacing

\bibitem{data}
\BIBentryALTinterwordspacing
J.~Basso, E.~Farhi, K.~Marwaha, B.~Villalonga, and L.~Zhou, ``{Performance of
  the QAOA on MaxCut over Large-Girth Regular Graphs},'' Online, 2022.
  \url{https://github.com/benjaminvillalonga/large-girth-maxcut-qaoa}
\BIBentrySTDinterwordspacing

\bibitem{eigenweb}
\BIBentryALTinterwordspacing
G.~Guennebaud, B.~Jacob \emph{et~al.}, ``Eigen v3,'' Online, 2010.
  \url{https://eigen.tuxfamily.org}
\BIBentrySTDinterwordspacing

\bibitem{dagum1998openmp}
\BIBentryALTinterwordspacing
L.~Dagum and R.~Menon, ``{OpenMP: an industry standard API for shared-memory
  programming},'' \emph{Computational Science \& Engineering, IEEE}, vol.~5,
  no.~1, pp. 46--55, 1998. \url{https://doi.org/10.1109/99.660313}
\BIBentrySTDinterwordspacing

\bibitem{lbfgspp}
\BIBentryALTinterwordspacing
Y.~Qiu and D.~Toewe, ``{LBFGS++},'' Online, 2020.
  \url{https://github.com/yixuan/LBFGSpp}
\BIBentrySTDinterwordspacing

\bibitem{panchenko2013sherrington}
\BIBentryALTinterwordspacing
D.~Panchenko, \emph{{The Sherrington-Kirkpatrick model}}.\hskip 1em plus 0.5em
  minus 0.4em\relax Springer Science \& Business Media, 2013.
  \url{https://doi.org/10.1007/978-1-4614-6289-7}
\BIBentrySTDinterwordspacing

\bibitem{schmidt2008replica}
\BIBentryALTinterwordspacing
M.~J. Schmidt, ``{Replica symmetry breaking at low temperatures},'' Ph.D.
  dissertation, Universit{\"a}t W{\"u}rzburg, 2008.
  \url{https://d-nb.info/991972910/34}
\BIBentrySTDinterwordspacing

\bibitem{auffinger2020RSB}
\BIBentryALTinterwordspacing
A.~Auffinger, W.-K. Chen, and Q.~Zeng, ``{The SK Model Is Infinite Step Replica
  Symmetry Breaking at Zero Temperature},'' \emph{Communications on Pure and
  Applied Mathematics}, vol.~73, no.~5, pp. 921--943, 2020.
  \href{https://arxiv.org/abs/1703.06872}{arXiv:1703.06872}
\BIBentrySTDinterwordspacing

\bibitem{SenHypergraph}
\BIBentryALTinterwordspacing
S.~Sen, ``Optimization on sparse random hypergraphs and spin glasses,''
  \emph{Random Structures \& Algorithms}, vol.~53, no.~3, pp. 504--536, 2018.
  \href{https://arxiv.org/abs/1606.02365}{arXiv:1606.02365}
\BIBentrySTDinterwordspacing

\bibitem{Berge}
\BIBentryALTinterwordspacing
C.~Berge, \emph{{Hypergraphs, Combinatorics of Finite Sets}}.\hskip 1em plus
  0.5em minus 0.4em\relax North-Holland, 1989.
  \url{http://compalg.inf.elte.hu/~tony/Oktatas/Algoritmusok-hatekonysaga/Berge-hypergraphs.pdf}
\BIBentrySTDinterwordspacing

\bibitem{BKKT}
\BIBentryALTinterwordspacing
S.~Bravyi, A.~Kliesch, R.~Koenig, and E.~Tang, ``Obstacles to variational
  quantum optimization from symmetry protection,'' \emph{Phys. Rev. Lett.},
  vol. 125, p. 260505, Dec 2020.
  \href{https://arxiv.org/abs/1910.08980}{arXiv:1910.08980}
\BIBentrySTDinterwordspacing

\bibitem{FGG20}
\BIBentryALTinterwordspacing
E.~Farhi, D.~Gamarnik, and S.~Gutmann, ``{The Quantum Approximate Optimization
  Algorithm Needs to See the Whole Graph: Worst Case Examples},'' \emph{arXiv
  preprint}, 2020. \href{https://arxiv.org/abs/2005.08747}{arXiv:2005.08747}
\BIBentrySTDinterwordspacing

\bibitem{FGG20typical}
\BIBentryALTinterwordspacing
------, ``{The Quantum Approximate Optimization Algorithm Needs to See the
  Whole Graph: A Typical Case},'' \emph{arXiv preprint}, 2020.
  \href{https://arxiv.org/abs/2004.09002}{arXiv:2004.09002}
\BIBentrySTDinterwordspacing

\bibitem{chou2021limitations}
\BIBentryALTinterwordspacing
C.-N. Chou, P.~J. Love, J.~S. Sandhu, and J.~Shi, ``{Limitations of Local
  Quantum Algorithms on Random Max-k-XOR and Beyond},'' \emph{arXiv preprint},
  2021. \href{https://arxiv.org/abs/2108.06049}{arXiv:2108.06049}
\BIBentrySTDinterwordspacing

\bibitem{BGMZ22}
\BIBentryALTinterwordspacing
J.~Basso, D.~Gamarnik, S.~Mei, and L.~Zhou, ``{Performance and limitations of
  the QAOA at constant levels on large sparse hypergraphs and spin glass
  models},'' \emph{arXiv preprint}, 2022.
  \href{https://arxiv.org/abs/2204.10306}{arXiv:2204.10306}
\BIBentrySTDinterwordspacing

\end{thebibliography}

\clearpage
\appendix

\section{Properties of the iterations}
\label{apx:iter-properties}

In this appendix, we prove some properties of the elements of the iterations. These properties are of interest in their own right and also necessary to fill in some gaps in the derivations.

We start in \cref{apx:f_identities} by proving some identities that will be used in the later proofs.
The assertion in Eq.~\eqref{eq:sum_f_H_in_main_text}  is  proved in \cref{apx:H_properties}.
We  prove in \cref{apx:proof_properties_G_m} the symmetry properties of the $G^{(m)}$ matrix elements asserted in \cref{sec:infinite_D_iteration}.
We show in \cref{apx:proof_placement_G} how each matrix element of $G^{(m)}$ only depends on a submatrix of $G^{(m-1)}$.
This is key to understanding the structure of the iteration and can be used to speed up the iteration by a factor of $p$.

\subsection{Properties of $f(\av)$}
\label{apx:f_identities}

We start by establishing some notation and proving a few properties of $f(\av)$ that will be used subsequently.
Let
\begin{equation}
B= \{ (a_1, a_2,\ldots, a_p, a_0, a_{-p}, \ldots, a_{-2}, a_{-1}) : a_j = \pm 1 \}
\end{equation}
be the set of $(2p+1)$-bit strings.
We define the following subset
\begin{equation}
B_0 = \{\av \in B : a_{-r} = a_r \text{ for all } 1\le r \le p\}.
\end{equation}

For any $\av \in B$, we define the $T(\av)$ to be the largest positive index $T$ such that $a_{T} \neq a_{-T}$ if $\av\not \in B_0$, and 0 if $\av \in B_0$.
More formally,
\begin{equation}
\label{eq:T-def}
T(\av) = 
\begin{cases}
\max \{r : a_{-r} \neq  a_r \} & \text{if } \av \not\in B_0 \\
0 & \text{if } \av \in B_0.
\end{cases}
\end{equation}

To see how $B$ is partitioned according to different $T(\av)$, see Table~\ref{table:T_prime_examples}.
Note that $a_{-r} = a_r$ whenever $r> T(\av)$.
These levels will help organize the $\av \in B$ in the proofs below. Moreover, for any $\av \not\in B_0$, let $\av' \not\in B_0$ be the following bit string
\begin{equation}\label{eq:prime_def}
a_{\pm r}' = \begin{cases}
a_0 & \text{ if } r = 0 \\
-a_{\pm r} &\text{ if } 1\le r \le T(\av) \\
a_{\pm r} &\text{ if } T(\av)+1  \le  r \le p.
\end{cases}
\end{equation}
As we will see, this definition is helpful as $\av$ and $\av'$ are going to pair up leading to cancellations.

\begin{table}[H]
\centering
\begin{tabular}{wl{0.5cm}wl{0.5cm}wl{0.5cm}wl{0.5cm}wr{0.5cm}wr{0.7cm}wr{0.7cm}|c|wr{0.85cm}wr{0.7cm}wr{0.7cm}wr{0.5cm}wr{0.5cm}wr{0.7cm}wr{0.7cm}}
\multicolumn{7}{c|}{$\av$} & $T(\av)$ &
\multicolumn{7}{c}{$\av'$} \\
\hline
$( a_{1 }$, & $a_{2 }$, & $a_{3 }$, & $ a_{0 }$, & $-a_{3 }$, & $a_{-2}$, & $a_{-1})$ &
$3$ 
& $(-a_{1 }$, & $-a_{2 }$, & $-a_{3 }$, & $ a_{0 }$, 
& $ a_{3 }$, & $-a_{-2}$, & $-a_{-1})$\\
$( a_{1 }$, & $a_{2 }$, & $a_{3 }$, & $a_{0 }$, 
& $a_{3 }$, & $-a_{2}$, & $a_{-1})$ & 
$2$ 
& $(-a_{1 }$, & $-a_{2 }$, & $ a_{3 }$, & $a_{0 }$, 
& $ a_{3 }$, & $a_{2 }$, & $-a_{-1})$\\ 
$( a_{1 }$, & $a_{2 }$, & $ a_{3 }$, & $a_{0 }$, 
& $a_{3 }$, & $ a_{2 }$, & $-a_{1 })$ &
$1$ 
& $(-a_{1 }$, & $a_{2 }$, & $a_{3 }$, & $a_{0 }$, 
& $a_{3 }$, & $a_{2}$, & $a_{1})$ \\ 
$( a_{1 }$, & $a_{2 }$, & $a_{3 }$, & $a_{0 }$, & $a_{3 }$, & $a_{2 }$, & $a_{1 })$ &
$0$ 
& $(\s a_{1 }$, & $a_{2 }$, & $a_{3 }$, & $a_{0 }$, & $a_{3 }$, & $   a_{2}$, & $a_{1})$
\end{tabular}
\caption{$\av$, $T(\av)$ and $\av'$ for $p=3$.}
\label{table:T_prime_examples}
\end{table}

We also define, for any $\av\in B$, a corresponding $\conj{\av} \in B$ whose entries are
\begin{equation}
    \conj{a}_j = a_{-j} \qquad\text{for } -p \le j \le p.
\end{equation}
That is, $\conj\av$ is the bit string $\av$ reversed.

Now we are ready to prove the identities on $f(\av)$. Recall its definition from Eq.~\eqref{eq:f_definition}:
\begin{align}\label{eq:f_def_appendix}
    f(\av) &= \frac{1}{2} \braket{a_1 | e^{i \beta_1 X} | a_2} \cdots \braket{ a_{p-1} | e^{i \beta_{p-1} X} | a_p} \braket{a_p | e^{i \beta_p X} | a_0} \nonumber \\
    &\quad \times \braket{a_{0} | e^{-i \beta_p X} | a_{-p}} \braket{a_{-p} | e^{-i \beta_{p-1} X} | a_{-(p-1)}} \cdots \braket{a_{-2} | e^{-i \beta_1 X} | a_{-1}},
\end{align}
which can be alternatively written as
\begin{align}\label{eq:eq:f_alternative}
     f(\av) &= \frac{1}{2} \braket{a_1 a_2| e^{i \beta_1 X} | 1} \cdots \braket{ a_{p-1} a_p | e^{i \beta_{p-1} X} | 1} \braket{a_p a_0 | e^{i \beta_p X} | 1} \nonumber \\
    &\quad \times \braket{a_{0} a_{-p}| e^{-i \beta_p X} | 1} \braket{a_{-p}a_{-(p-1)} | e^{-i \beta_{p-1} X} | 1} \cdots \braket{a_{-2} a_{-1}| e^{-i \beta_1 X} | 1}
\end{align}
where
\begin{equation}
\braket{a|e^{i\beta X} |1} = 
\begin{cases}
\cos(\beta) &\text{ if } a = +1 \\
i\sin(\beta) &\text{ if } a = -1 \,.
\end{cases} 
\end{equation}

\begin{lemma} \label{id:f-prime}
For any $\av \in B\setminus B_0$,
\begin{equation}
f(\av') = - f(\av).
\end{equation}
\end{lemma}

\begin{proof}

Note that the set $B$ can be subdivided into sets of strings $\{\av\}$ with fixed values of $T(\av)$. We will start with the set of $\av$'s where $T(\av)=p$.
Note from Eq.~\eqref{eq:eq:f_alternative} that $f(\av)$ only depends on the product of pairs of adjacent bits.
Under the prime operation $\av\to\av'$,
these products will not change except for $a_p a_0$ and $a_0 a_{-p}$. (As an example look at $T=3$ in Table~\ref{table:T_prime_examples}.)
Flipping the signs of these two products flips the sign of the product of the two relevant matrix elements in Eq.~\eqref{eq:eq:f_alternative}.
Therefore, in this case, $f(\av') = - f(\av)$.

Now consider the set of $\av$'s with $T(\av)=p-1$. Now the product of any pair of adjacent bits will not change under the prime operation except for $a_{p-1} a_p$ and $a_{-p} a_{-(p-1)}$. (As an example look at $T=2$ in Table~\ref{table:T_prime_examples}.) Flipping the signs of these two products flips the sign of the product of the two relevant matrix elements in Eq.~\eqref{eq:eq:f_alternative}.
Therefore in this case again we have that $f(\av') = - f(\av)$. This argument can be repeated for 
$T(\av)=p-2, p-3, ..., 1$
 so the Lemma is established.
\end{proof}

\begin{lemma}\label{id:f-conj}
For any $\av \in B$,
\begin{equation}\label{eq:f_conj_equation}
    f(\conj\av) = f(\av)^*.
\end{equation}
\end{lemma}

\begin{proof}
This follows from Eq.~\eqref{eq:f_def_appendix} and noting that
\begin{align*}
\braket{a_1 | e^{i \beta_1 X} | a_2} \braket{a_{-2} | e^{-i \beta_1 X} | a_{-1}} &= \braket{\conj{a}_{-1} | e^{i \beta_1 X} | \conj{a}_{-2}} \braket{\conj{a}_{2} | e^{-i \beta_1 X} | \conj{a}_{1}} \nonumber \\
&= \braket{\conj{a}_{1} | e^{i \beta_1 X} | \conj{a}_{2}}^* \braket{\conj{a}_{-2} | e^{-i \beta_1 X} | \conj{a}_{-1}}^*.
\end{align*}
The same follows for the other pairs, yielding Eq.~\eqref{eq:f_conj_equation}.
\end{proof}

\begin{lemma} \label{id:fsum}
\begin{equation} \label{eq:fsum}
\sum_{\av \in B} f(\av) = \sum_{\av \in B_0} f(\av) = 1.
\end{equation}
\end{lemma}
\begin{proof}
Using definition of $f(\av)$ in \cref{eq:f_def_appendix}, we see that
\begin{equation*}
 \sum_{\av \in B} f(\av) = \sum_{a_1, a_{-1}} \frac12 \braket{a_1| e^{i \beta_1 X} \cdots e^{i \beta_p X} e^{-i \beta_p X} \cdots e^{-i \beta_1 X} | a_{-1}} = \sum_{a_1, a_{-1}} \frac12 \braket{a_1| a_{-1}}  = 1.
\end{equation*}
Also
\begin{equation*}
    \sum_{\av \in B} f(\av) = \sum_{\av \in B_0} f(\av)  + \frac12 \sum_{\av \not\in B} \big[f(\av) + f(\av')\big] = \sum_{\av \in B_0} f(\av)
\end{equation*}
where we have decomposed the sum over $\av\in B$ into a sum over $\av\in B_0$ and a sum over $\av \not \in B_0$, and duplicated the latter by considering summing over $\av'$ instead of $\av$. Then applying \cref{id:f-prime} gives the first equality in \cref{eq:fsum}.  
\end{proof}

\subsection{Properties of $H^{(m)}_D(\av)$}\label{apx:H_properties}
\label{apx:H-properties}

We now prove some properties involving $H^{(m)}_D(\av)$, which are used in the proofs in Section~\ref{sec:iter_proofs_infinite} and in the later sections of the appendix. 

\begin{lemma} \label{lem:H-D-prop}
For $0 \leq m \leq p$,
\begin{align}
H_D^{(m)}(\av) = 1 
\quad \textnormal{ if } \quad
\av \in B_0
\qquad \textnormal{ and } \qquad
H_D^{(m)}(\av') = H_D^{(m)}(\av) 
\quad \textnormal{ if } \quad\av \not\in B_0.
\end{align}
\end{lemma}
\begin{proof}
We proceed by induction on $m$. The base case $m=0$ follows from the fact that, for all $\av$, $H_D^{(0)}(\av) = 1$.  Now as an inductive hypothesis assume that the lemma is true for $m-1$.
Then, using \cref{id:f-prime}, we can break the sum on $\bv$ in the definition \eqref{eq:H_m_def} of $H_D^{(m)}(\av)$ into pieces as
\begin{align}
    H_D^{(m)}(\av) &= \bigg(\sum_{\bv:~ T(\bv) \le T(\av)} f(\bv) H_D^{(m-1)}(\bv) \cos{\Big[ {\textstyle \frac{1}{\sqrt{D}} }\Gammav \cdot (\av \bv) \Big]} \nonumber \\
    &\quad \quad + \frac12 \sum_{\bv :~ T(\bv) > T(\av)} f(\bv) H_D^{(m-1)}(\bv) \bigg[ \cos{\Big[ {\textstyle \frac{1}{\sqrt{D}}}\Gammav \cdot (\av \bv) \Big]} - \cos{\Big[ {\textstyle \frac{1}{\sqrt{D}} }\Gammav \cdot (\av \bv') \Big]} \bigg] \bigg)^D . \label{eq:H_D_expanded_appendix}
\end{align}
We now show that the second sum evaluates to zero.
To see this, note that
\begin{align}
\label{eq:gamma_av_bv_flips_sign_with_prime}
    \Gammav \cdot (\av\bv) = \sum_{r=1}^p \gamma_r (a_r b_r - a_{-r} b_{-r}) = \sum_{r=1}^{\max\{T(\av), T(\bv)\} } \gamma_r (a_r b_r - a_{-r} b_{-r})
\end{align}
since $a_r b_r = a_{-r} b_{-r}$ when $r > \max\{T(\av), T(\bv)\}$.
Note if $T(\bv) \ge T(\av)$, then $b_{\pm r} = - b'_{\pm r}$ when $r \leq \max\{T(\av), T(\bv)\} = T(\bv) = T(\bv')$. So whenever $T(\bv) \ge T(\av)$, which includes the case when $T(\bv) > T(\av)$, we have from Eq.~\eqref{eq:gamma_av_bv_flips_sign_with_prime} that
\begin{align}\label{eq:Gamma_ab_becomes_prime}
    \Gammav \cdot (\av\bv) = -\sum_{r=1}^{\max\{T(\av), T(\bv')\}} \gamma_r (a_r b'_r - a_{-r} b'_{-r}) 
    = - \Gammav \cdot (\av\bv').
\end{align}
And since cosine is an even function, Eq.~\eqref{eq:H_D_expanded_appendix} becomes
\begin{align}
\label{eq:H_D_simplified_levels}
H_D^{(m)}(\av) &= \bigg(\sum_{\substack{\bv :~ T(\bv) \leq T(\av)}} f(\bv) H_D^{(m-1)}(\bv) \cos{\Big[{\textstyle \frac{1}{\sqrt{D}}} \Gammav \cdot (\av \bv) \Big]} \bigg)^D .
\end{align}

With this simplified form of $H_D^{(m)}(\av)$, we now can prove the Lemma.
If $\av \in B_0$, we have $T(\av) = 0$.
Using the inductive hypothesis that $H^{(m-1)}_D(\bv) = 1$ when $\bv\in B_0$ and the fact that $\Gammav \cdot (\av\bv) = 0$ when $\av,\bv \in B_0$, we have
\begin{equation}
H_D^{(m)}(\av) = \bigg(\sum_{\bv \in B_0} f(\bv) H_D^{(m-1)}(\bv) \cos{\Big[ {\textstyle \frac{1}{\sqrt{D}}}\Gammav \cdot (\av \bv) \Big]} \bigg)^D = \bigg(\sum_{\bv \in B_0} f(\bv) \bigg)^D = 1,
\end{equation}
where \cref{id:fsum} is used in the last equality.
Now consider the case of $\av \not \in B_0$.
When $T(\bv) \le T(\av)$ as in the sum in Eq.~\eqref{eq:H_D_simplified_levels},
we may switch the roles of $\av,\bv$ in Eq.~\eqref{eq:Gamma_ab_becomes_prime} to see that $\Gammav \cdot (\av\bv) = - \Gammav \cdot (\av'\bv)$.
Since $T(\av) = T(\av')$, we can rewrite Eq.~\eqref{eq:H_D_simplified_levels} as
\begin{align}
    H_D^{(m)}(\av) &= \bigg(\sum_{\substack{\bv :~ T(\bv) \leq T(\av')}} f(\bv) H_D^{(m-1)}(\bv) \cos{\Big[ {\textstyle \frac{1}{\sqrt{D}}}\Gammav \cdot (\av' \bv) \Big]} \bigg)^D =H_D^{(m)}(\av').
\end{align}
This completes the induction.
\end{proof}

Next, we apply the above result to prove the assertion of Eq.~\eqref{eq:sum_f_H_in_main_text}, which is used in \cref{sec:iter_proofs_infinite} to establish the correctness of the infinite-$D$ iteration.
This assertion is the following lemma that we now prove.

\begin{lemma}\label{lem:fH-sum}
For $0\le m\le p$, we have
\begin{equation}
    \sum_{\av \in B} f(\av) H_D^{(m)}(\av) = 1 .
\end{equation}
\end{lemma}
\begin{proof}
Decomposing the sum over $\av\in B$ and applying  Lemmas~\ref{id:f-prime} and \ref{lem:H-D-prop}, we get
\begin{align*}
\sum_{\av\in B} f(\av) H_D^{(m)}(\av) &= \sum_{\av \in B_0} f(\av) H_D^{(m)}(\av) 
+ \frac12 \sum_{\av \not\in B_0} \Big[f(\av) H_D^{(m)}(\av)  + f(\av') H_D^{(m)}(\av') \Big] \nonumber \\
&= \sum_{\av \in B_0} f(\av) 
+ \frac12 \sum_{\av \not\in B_0} \Big[f(\av) H_D^{(m)}(\av)  - f(\av) H_D^{(m)}(\av) \Big] \nonumber \\
 &= \sum_{\av \in B_0} f(\av) = 1
\end{align*}
where we used \cref{id:fsum} on the last line.
\end{proof}

Finally, we prove a lemma to be used in proving the properties of $G^{(m)}_{r,s}$ that involve complex conjugation in what follows in Appendix~\ref{apx:proof_properties_G_m}.
\begin{lemma}\label{lem:H_conj}
For any $\av \in B$ and $0 \leq m \leq p$,
\begin{equation}\label{eq:H_conj}
H_D^{(m)}(\conj\av) = H_D^{(m)}(\av)^*.
\end{equation}
\end{lemma}
\begin{proof}
We proceed by induction on $m$. The base case $m=0$ follows from the fact that $H^{(0)}_D(\av) = 1$. Now assume that the lemma is true for $m-1$. By Eq.~\eqref{eq:H_m_def},
\[
    H_D^{(m)}(\conj\av) = \bigg(\sum_{\bv} f(\bv) H_D^{(m-1)}(\bv) \cos{\Big[ {\textstyle \frac{1}{\sqrt{D}}}\Gammav \cdot (\conj\av \bv) \Big]} \bigg)^D = \bigg(\sum_{\bv} f(\conj\bv) H_D^{(m-1)}(\conj\bv) \cos{\Big[ {\textstyle \frac{1}{\sqrt{D}}}\Gammav \cdot (\conj\av \conj\bv) \Big]} \bigg)^D
\]
since $\sum_{\av \in B} (\cdots) = \sum_{\conj\av \in B} (\cdots)$. Then, by the inductive hypothesis, Lemma 2 and  the fact  that $\Gammav \cdot(\conj\av \conj\bv) = -\Gammav \cdot (\av \bv)$ we have
\begin{equation*}
    H_D^{(m)}(\conj\av) = \bigg(\sum_{\bv} f(\bv)^* H_D^{(m-1)}(\bv)^* \cos{\Big[ {\textstyle \frac{1}{\sqrt{D}}}\Gammav \cdot (\av \bv) \Big]} \bigg)^D = H_D^{(m)}(\av)^*.
\end{equation*}
\end{proof}

\clearpage

\subsection{Symmetries of the $G^{(m)}$ matrix}
\label{apx:proof_properties_G_m}

In this section of the appendix, we prove the following symmetry properties of the matrix $G^{(m)}$ that are stated in \cref{sec:infinite_D_iteration}:
For $1\le r < s \le p$, and $-p\le j, k \le p$, we have
\begin{multicols}{2}
\begin{enumerate}[label=(\arabic*),leftmargin=3\parindent]
    \item \label{item:prop_1} $G_{j,k}^{(m)} = G_{k,j}^{(m)}$,
    
    \item \label{item:prop_2} $G_{j,j}^{(m)} = G_{j,-j}^{(m)} = 1$,
    
    \item \label{item:prop_3} $G_{0, r}^{(m)} = G_{0,-r}^{(m)*}$,
    
    \item \label{item:prop_4} $G_{r,s}^{(m)} = G_{r,-s}^{(m)}= G_{-r,-s}^{(m)*} = G_{-r,s}^{(m)*}$.
\end{enumerate}
\end{multicols}

Let us use the definition of $G^{(m)}_{j,k}$ in terms of $H^{(m)}$ in Eq.~\eqref{eq:G_definition} which we reproduce here:
\begin{equation}
    G_{j,k}^{(m)} := \sum_{\av \in B} f(\av) H^{(m)}(\av) a_j a_k . \label{eq:G_def_appendix}
\end{equation}
It is clear from this definition that $G^{(m)}$ is a symmetric matrix, so Property~\ref{item:prop_1} holds. For the remaining properties, let us rewrite \cref{eq:G_def_appendix} using Lemmas~\ref{id:f-prime} and \ref{lem:H-D-prop} as
\begin{equation} \label{eq:G-decomposed}
G_{j,k}^{(m)} = 
    \sum_{\av \in B_0} f(\av) a_j a_k + \frac12\sum_{\av\not\in B_0} f(\av) H^{(m)}(\av) (a_j a_k - a'_j a'_k).
\end{equation}
Note that Lemma~\ref{lem:H-D-prop} is applied to $H^{(m)}(\av)$ which is $H^{(m)}_D(\av)$ in the $D \to \infty$ limit.

To prove Property~\ref{item:prop_2}, observe that 
\begin{equation}
    G_{j,j}^{(m)} = \sum_{\av \in B_0} f(\av)  + \frac12\sum_{\av\not\in B_0} f(\av) H^{(m)}(\av) (1 - 1) = \sum_{\av \in B_0} f(\av)  =1, 
\end{equation}
where the last equality is due to \cref{id:fsum}.
Also
\begin{equation*}
    G_{j,-j}^{(m)} = \sum_{\av \in B_0} f(\av) a_j a_{-j} + \frac12\sum_{\av\not\in B_0} f(\av) H^{(m)}(\av) (a_j a_{-j} - a'_{j} a'_{-j}).
\end{equation*}
Note that $a_j=a_{-j}$ for $\av\in B_0$, and $a_j a_{-j} = a'_j a'_{-j}$ for any $\av \in B$. Therefore
\begin{equation}
    G_{j,-j}^{(m)} =\sum_{\av \in B_0} f(\av)  = 1.
\end{equation}

Next, we prove Property~\ref{item:prop_3}.
From \cref{eq:G_def_appendix} we have
\begin{equation*}
    G_{0,r}^{(m)} = \sum_{\av \in B} f(\av) H^{(m)}(\av) a_0 a_r = \sum_{\av \in B} f(\conj\av) H^{(m)}(\conj\av) \conj{a}_0 \conj{a}_r
\end{equation*}
where we used the fact that $\sum_{\av \in B} (\cdots) = \sum_{\conj\av \in B} (\cdots)$.
Applying \cref{lem:H_conj}, which is true for each $D$ so it is true in the $D \to \infty$ limit, along with \cref{id:f-conj} we get
\begin{equation}
    G_{0,r}^{(m)} = \sum_{\av \in B} f(\av)^* H^{(m)}(\av)^* a_0 a_{-r} = G_{0,-r}^{(m)*} .
\end{equation}

Lastly, we need to prove Property~\ref{item:prop_4}. Let $1\le r < s \le p$. Then from \cref{eq:G-decomposed} we have
\begin{align*}
G_{r,s}^{(m)} &= \sum_{\av\in B_0} f(\av) a_r a_s
+ \frac12\sum_{\av\not\in B_0} f(\av)H^{(m)}(\av) ( a_r a_s  - a'_r a'_s ) \nonumber \\
&=\sum_{\av\in B_0} f(\av) a_r a_s
+ \frac12\sum_{\av\not\in B_0, r \le T(\av) <s} f(\av)H^{(m)}(\av) ( a_r a_s  - a'_r a'_s )
\end{align*}
where we used the fact that the second summand is nonzero only if $a_r a_s \neq a'_r a'_s$, which holds only if $r \le T(\av) < s$.
Under this condition we have $a_r a_s - a'_r a'_s = 2a_r a_s$, so
\begin{align}\label{eq:G_r_s_expanded}
G_{r,s}^{(m)}=\sum_{\av\in B_0} f(\av) a_r a_s
+ \sum_{\av\not\in B_0, r \le T(\av) <s} f(\av)H^{(m)}(\av) a_r a_s .
\end{align}
Similarly,
\begin{align}
G_{r,-s}^{(m)} &= \sum_{\av\in B_0} f(\av) a_r a_{-s}
+ \sum_{\av\not\in B_0, r \le T(\av) < s} f(\av)H^{(m)}(\av) a_r a_{-s}  \nonumber \\
&= \sum_{\av\in B_0} f(\av) a_r a_{s}
+ \sum_{\av\not\in B_0, r \le T(\av) < s} f(\av)H^{(m)}(\av) a_r a_{s} \label{eq:G_r_-s}
\end{align}
where we used the fact that $a_{-s} = a_{s}$ if $\av \in B_0$ or if $T(\av) < s$. Comparing Eq.~\eqref{eq:G_r_s_expanded} and Eq.~\eqref{eq:G_r_-s}, we conclude that
\begin{equation}
    G_{r,s}^{(m)} = G_{r,-s}^{(m)} \,.
\end{equation}
It remains to consider
\begin{equation*}
G_{-r,\pm s}^{(m)} = \sum_{\av \in B} f(\av) H(\av) a_{-r} a_{\pm s} = \sum_{\av \in B} f(\conj\av) H(\conj\av) \conj{a}_{-r} \conj{a}_{\pm s}
\end{equation*}
where we used the fact that $\sum_{\av \in B} (\cdots) = \sum_{\conj\av \in B} (\cdots)$. Now, using \cref{id:f-conj} and \cref{lem:H_conj} in the $D \rightarrow \infty$ limit,
\begin{equation}
    G_{-r,\pm s}^{(m)} = \sum_{\av \in B} f(\av)^* H(\av)^* a_{r} a_{\mp s} = G_{r, \mp s}^{(m)^*}
\end{equation}
establishing Property~\ref{item:prop_4}.

\clearpage

\subsection{Placement of the matrix elements of $G^{(m)}$}
\label{apx:proof_placement_G}

In this section of the appendix, we show how the matrix elements of $G^{(m)}$ only depend on a submatrix of $G^{(m-1)}$.
This implies that the iteration on $G^{(m)}$ can be understood as a placement of the entries in the matrix $G^{(m)}$ that remain unchanged after a sufficient number of steps.
We also elaborate on why this enables us to reduce the complexity of implementing the iteration in \cref{sec:infinite_D_iteration} to $O(p^2 4^p)$ instead of the na\"ive $O(p^3 4^p)$ estimate.

Recall the definition \eqref{eq:G_definition} of $G^{(m)}_{j,k}$ which states
\begin{equation} \label{eq:G_def_appendix_A4}
    G_{j,k}^{(m)} = \sum_{\av} f(\av) H^{(m)}(\av) a_j a_k.
\end{equation}
Observe that combining Eqs.~\eqref{eq:H_tilde_m_grouped} and \eqref{eq:G_definition} we have
\begin{equation} \label{eq:Hm-from-G}
    H^{(m)}(\av) =
    \exp\Big[{-}
    \frac12\sum_{j',k'=-p}^p G_{j',k'}^{(m-1)} \Gamma_{j'} \Gamma_{k'} a_{j'} a_{k'}
    \Big].
\end{equation}
Note the summand in the exponential is zero whenever either $j'=0$ or $k'=0$ since $\Gamma_0=0$.
Hence, we can focus on $j',k'\neq 0$, and enumerate all pairs of $(j',k')$ by writing the sum as
\begin{align}
H^{(m)}(\av)
& = \exp\bigg\{ {-} \frac12 \sum_{s'=1}^p \Big[
	G_{s',s'}^{(m-1)} \gamma_{s'}^2 - G_{s',-s'}^{(m-1)} \gamma_{s'}^2 a_{s'} a_{-s'} - G_{-s',s'}^{(m-1)} \gamma_{s'}^2 a_{-s'} a_{s'} + G_{-s',-s'}^{(m-1)} \gamma_{s'}^2 
	\Big]
	\nonumber \\ 
&\qquad \quad~~	-\sum_{1 \le r' < s'\le p} 
	\Big[ G_{r',s'}^{(m-1)} \gamma_{r'} \gamma_{s'} a_{r'} a_{s'} - G_{r',-s'}^{(m-1)} \gamma_{r'} \gamma_{s'} a_{r'} a_{-s'} \nonumber \\
& \qquad \qquad \qquad \qquad \quad 
	- G_{-r',s'}^{(m-1)} \gamma_{r'} \gamma_{s'} a_{-r'} a_{s'} + G_{-r',-s'}^{(m-1)} \gamma_{r'} \gamma_{s'} a_{-r'} a_{-s'}
	\Big] \bigg\}
\end{align}
where the factor of $1/2$ in front of the second sum is killed by considering the contribution from the symmetric pair $G_{s',r'}^{(m-1)}=G_{r',s'}^{(m-1)}$, etc.
Then applying the Properties~\ref{item:prop_2} and \ref{item:prop_4} listed in \cref{apx:proof_properties_G_m} to $G^{(m-1)}$ yields the following simplified form
\begin{equation} \label{eq:Hm-simplified}
H^{(m)}(\av)
= \exp\bigg\{
{-} \sum_{s'=1}^p \gamma_{s'}^2 (1 -  a_{s'} a_{-s'})
	-\hspace{-2pt} \sum_{1 \le r' < s'\le p} \gamma_{r'} \gamma_{s'}
	\Big[ G_{r',s'}^{(m-1)}  a_{r'} 
		- G_{r',s'}^{(m-1)*}  a_{-r'} 
	\Big](a_{s'} - a_{-s'})	
	\bigg\}.
\end{equation}

Going back to Eqs.~\eqref{eq:G_def_appendix_A4} and \eqref{eq:Hm-from-G} we see that $G^{(m)}$ only depends on the matrix elements of $G^{(m-1)}$ through $H^{(m)}(\av)$, which only involves $G^{(m-1)}_{r',s'}$ for $1\le r' < s' \le p$ as seen from the above equation.
In particular this means that $G^{(m)}_{0,r}$ depends only on $G^{(m-1)}_{r',s'}$ for $1\le r' < s' \le p$ .
 
We next show that for $1\le r < s \le p$,
$G^{(m)}_{r, s}$ only depends on $G^{(m-1)}_{r', s'}$ where $1 \leq r' < s'< s$.
We first restate \cref{eq:G_r_s_expanded} 
\begin{align}
\label{eq:Grs-place}
   G_{r,s}^{(m)} = \sum_{\av\in B_0} f(\av) a_r a_s
+ \sum_{\av\not\in B_0, r \le T(\av) <s} f(\av) H^{(m)}(\av) a_r a_s\, .
\end{align}
Observe that the dependence of $G^{(m)}_{r,s}$ on the matrix elements of $G^{(m-1)}$ is only via $H^{(m)}(\av)$ in the second summand in Eq.~\eqref{eq:Hm-simplified}. 
However, this summand can be nonzero only if $a_{s'} \neq a_{-s'}$, which holds only if $T(\av) \ge s'$. 
Additionally, in Eq.~\eqref{eq:Grs-place}, $G^{(m)}_{r,s}$ only depends on $H^{(m)}(\av)$ when $\av$ satisfies $T(\av) < s$.
Hence the only relevant terms in Eq.~\eqref{eq:Hm-simplified} involve $r'<s' \le T(\av) < s$.
Therefore, $G^{(m)}_{r,s}$ only depends on $G^{(m-1)}_{r',s'}$ where $r'<s'<s$.

\makeatletter
\renewcommand{\ALG@name}{Iteration}
\makeatother
\renewcommand{\thealgorithm}{}

\paragraph{Improved complexity of the iteration.}---
We now show why the aforementioned properties imply that we can improve the time complexity of the iteration on $G^{(m)}$ in \cref{sec:infinite_D_iteration} from the na\"ive $O(p^3 4^p)$ estimate to $O(p^2 4^p)$.

Given the dependency of the elements of $G^{(m)}$ on those of $G^{(m - 1)}$, at the $m$-th step of the iteration, when $1 \leq m \leq p - 1$, we only need to place the elements $G^{(m)}_{r, m+1}$, with $1\le r < m+1$.
Looking at \cref{eq:Grs-place} we see that calculating $G^{(m)}_{r, m+1}$ only involves summing over $\av$ where $\av \in B_0$ or $r\le T(\av) < m+1$.
So we can use the following implementation of our iteration:
\begin{algorithm}[H]
\caption{for computing $\nu_p(\paramv)$}
\begin{algorithmic}[1]
\State Allocate memory for matrices $G^{(0)}, G^{(1)}\in \mathds{C}^{(2p+1)\times(2p+1)}$, where the diagonal and anti-diagonal entries are initialized to $1$ due to Property~\ref{item:prop_2}, and the rest initialized to 0.

\For{$m=1,2,\ldots, p$}
    \For{$\av \in B_0 \cup \{\av: 1 \le T(\av) < m+1\}$}
        \State Compute $f(\av)$
            \Comment{Use \cref{eq:f_definition}}
        \State Compute $H^{(m)}(\av)$
            \Comment{Use \cref{eq:Hm-simplified} restricted to $s' \le T(\av) < m+1$}
        \For{$r = 1,2,\ldots, m$}
            \If{$m \le p-1$} \Comment{This implements \cref{eq:Grs-place}}
                \State Update $G^{(m)}_{r,m+1} \gets G^{(m)}_{r,m+1} + f(\av) H^{(m)}(\av) a_r a_{m+1}$ when $\av \in B_0$ or $T(\av)\ge r$

            \Else 
                \State
                Update $G^{(p)}_{0,r} \gets G^{(p)}_{0,r} + f(\av) H^{(p)}(\av) a_0 a_r$
            \EndIf
        \EndFor
    \EndFor
    \State If $m\le p-1$ then remove $G^{(m-1)}$ from memory and create a copy $G^{(m+1)} \gets G^{(m)}$.
\EndFor	
\State Returns $\nu_p(\paramv) =  (i/2) \sum_{r=1}^p \gamma_r [(G_{0,r}^{(p)})^2 - (G_{0,r}^{(p)*})^2]$

\Comment{Use \cref{eq:infinite_D_iter_result} and $G$'s symmetry properties}
\end{algorithmic}
\end{algorithm}

Let us evaluate the complexity of the above iteration.
Note there are $2^{p+1}$ elements in $B_0$, and $2^{p+\ell}$ elements $\av\in B$ satisfying $T(\av)=\ell \ge 1$. 
(It may be helpful to see this via Table~\ref{table:T_prime_examples}.)
So in line 3 above we only iterate over $O(2^{p+m})$ many $\av$'s.
Lines 4 and 5 make use of the fact that the factors $f(\av)$ and $H^{(m)}(\av)$ are the same regardless of which element $G^{(m)}_{r,m+1}$ we are computing.
They require only $O(p)$ and $O(m^2)$ time, respectively. Thus, the overall time complexity of the iteration is
\[
    O\Big({\textstyle \sum_{m=1}^p 2^{p+m}(p+m^2+m)}\Big) = O(p^2 4^p).
\]
The memory complexity is at most $O(p^2)$ for storing the $2$ matrices of dimension $(2p+1)\times (2p+1)$.

\clearpage

\section{A lemma needed in \cref{sec:SK_tree_agreement}}
\label{apx:equiv}

In this appendix, we prove Eq.~\eqref{eq:H-bar-prop} that was asserted in \cref{sec:SK_tree_agreement} and used to prove \cref{thm:SKequiv}.
This assertion turns out to be a version of Lemma~\ref{lem:H-D-prop} in the $D\to\infty$ limit, except we need to bridge the current formalism using $(2p+1)$-bit strings with the formalism in Ref.~\cite{farhi2019quantum} using $2p$-bit strings.

Following the notation in Ref.~\cite{farhi2019quantum}, we let 
\begin{align}
A = \{(a_1, a_2, \ldots, a_p, a_{-p}, \ldots, a_2, a_1) : a_i =\pm 1 \}    
\end{align}
be the set of $2p$-bit strings. We also let
\begin{align}
    A_{p+1} = \{\av \in A: a_{-r} = a_r \text{ for all } 1\le r \le p\}.
\end{align}
Note that $A_{p+1}$ is the set $B_0$ with the $0$-th component removed.  We use the index $p+1$ because it matches the conventions of Ref.~\cite{farhi2019quantum}. 
Now note that $T(\av)$ as defined in \cref{eq:T-def} for $\av \in B$ can also take $\av \in A$ as argument since it does not depend on the value of $a_0$.
So in what follows we slightly abuse notation to let $T(\av)$ to take arguments of both $\av \in A$ and $\av \in B$.

We now define the ``bar'' operation from Ref.~\cite{farhi2019quantum} as mentioned in \cref{sec:SK_tree_agreement}.
For any $\av \in A$,
we let
\begin{align}\label{eq:bar_def}
\bar{a}_{\pm r} = \begin{cases}
-a_{\pm r} & \text{ if } r = T(\av) \\
a_{\pm r} & \text{ if } r \neq T(\av)\,.
\end{cases}
\end{align}
We also need the * operation defined in \cref{eq:star_def}, which we restate here: for any $\av \in A$, let $\av^*$ be the bit string whose bits are given by
\begin{align}\label{eq:star_def_again2}
a^*_r=a_r a_{r+1}\cdots a_p
\qquad \text{and} \qquad
a^*_{-r} = a_{-r} a_{-r-1}\cdots a_{-p}
\qquad
\text{for } 1\le r\le p \,.
\end{align}
The * operation is one-to-one and therefore has a well defined inverse acting on $2p$-bit strings.

As discussed in \cref{sec:SK_tree_agreement}, while $H^{(m)}(\av)$ is defined for $\av \in B$, it does not depend on $a_0$. So we abuse notation to let it take arguments of both $\av \in A$ and $\av\in B$.
The assertion of \cref{eq:H-bar-prop} we wish to prove is the following lemma:
\begin{lemma} \label{lem:H-bar-prop}
For any $\av \in A$ and $0 \leq m \leq p$, we have
\begin{align} \label{eq:H-bar-prop-again2}
H^{(m)}(\av^*) = 1 \quad \textnormal{ if  } \quad \av \in A_{p+1} \qquad
\textnormal{and} \qquad
H^{(m)}(\bar{\av}^*) = H^{(m)}(\av^*)
\quad \textnormal{ if } \quad \av \not\in A_{p+1}
\end{align}
\end{lemma}
\begin{proof}
We start by showing the first half of \cref{eq:H-bar-prop-again2}.
Note that if $\av \in A_{p+1}$ then $\av^* \in A_{p+1}$.
Also note that $H^{(m)}(\bv)$ does not depend on $b_0$.
Then for any $\av \in A_{p+1}$, there is a corresponding $\bv \in B_0 $  given by $b_{\pm r} = a^*_{\pm r}$ for $1\le r \le p$,  and $b_0 = 1$ so that $H^{(m)}(\av^*) = H^{(m)}(\bv)$. Applying \cref{lem:H-D-prop} in the $D\to\infty$ limit with $\bv\in B_0$,
we have $H^{(m)}(\av^*) = H^{(m)}(\bv) = 1$.

For the second half of \cref{eq:H-bar-prop-again2}, we note from Eqs.~\eqref{eq:bar_def} and \eqref{eq:star_def_again2} that for any $\av \in A$,
\begin{align}
    \bar{a}^*_{\pm r} =
    \begin{cases}
    -a^*_{\pm r} & \quad \text{ if } 1 \leq r \leq T(\av) \\
    a^*_{\pm r} & \quad \text{ if } T(\av) + 1 \leq r \leq p \,.\\
    \end{cases}
\end{align}
This is the same as the definition of the prime operation on $\bv \in B$ in Eq.~\eqref{eq:prime_def}, if one ignores the $0$-th component.
When $b_{\pm r} = a^*_{\pm r}$,  we have $b'_{\pm r} = \bar{a}^*_{\pm r}$. 
Hence, for any $\av \in A\setminus A_{p+1}$, there is a corresponding $\bv \in B\setminus B_0$ with $b_0 = 1$,
such that
$H^{(m)}(\av^*) = H^{(m)}(\bv)$ and $H^{(m)}(\bar\av^*) = H^{(m)}(\bv')$.
Since $H^{(m)}(\bv') = H^{(m)}(\bv)$ by \cref{lem:H-D-prop} in the $D\to\infty$ limit, we have
$H^{(m)}(\bar\av^*) = H^{(m)}(\av^*)$, concluding the proof.
\end{proof}

\section{Tables of best known $\paramv$ for MaxCut \label{apx:optimal_angles}}

\begin{table}[H]
\centering
\begin{tabular}{c|c|c}
$p$ & $\gammav$ &  $\betav$ \\
\hline \hline
1 & $1 / 2$ &  $\pi / 8$ \\
 \hline
2 & 0.3817, 0.6655 &  0.4960, 0.2690 \\
 \hline
3 & 0.3297, 0.5688, 0.6406 &  0.5500, 0.3675, 0.2109 \\
 \hline
4 & 0.2949, 0.5144, 0.5586, 0.6429 &  0.5710, 0.4176, 0.3028, 0.1729 \\
 \hline
5 & 0.2705, 0.4804, 0.5074, 0.5646, 0.6397 &  0.5899, 0.4492, 0.3559, 0.2643, 0.1486 \\
 \hline
6 & 0.2528, 0.4531, 0.4750, 0.5146, 0.5650, &  0.6004, 0.4670, 0.3880, 0.3176, 0.2325, \\
& 0.6392 &  0.1291 \\
 \hline
7 & 0.2383, 0.4327, 0.4516, 0.4830, 0.5147, &  0.6085, 0.4810, 0.4090, 0.3534, 0.2857, \\
& 0.5686, 0.6393 &  0.2080, 0.1146 \\
 \hline
8 & 0.2268, 0.4162, 0.4332, 0.4608, 0.4818, &  0.6151, 0.4906, 0.4244, 0.3780, 0.3224, \\
& 0.5179, 0.5717, 0.6393 &  0.2606, 0.1884, 0.1030 \\
 \hline
9 & 0.2172, 0.4020, 0.4187, 0.4438, 0.4592, &  0.6196, 0.4973, 0.4354, 0.3956, 0.3481, \\
& 0.4838, 0.5212, 0.5754, 0.6398 &  0.2973, 0.2390, 0.1717, 0.0934 \\
 \hline
10 & 0.2089, 0.3902, 0.4066, 0.4305, 0.4423, &  0.6235, 0.5029, 0.4437, 0.4092, 0.3673, \\
& 0.4604, 0.4858, 0.5256, 0.5789, 0.6402 &  0.3246, 0.2758, 0.2208, 0.1578, 0.0855 \\
 \hline
11 & 0.2019, 0.3799, 0.3963, 0.4196, 0.4291, &  0.6268, 0.5070, 0.4502, 0.4195, 0.3822, \\
& 0.4431, 0.4611, 0.4895, 0.5299, 0.5821, &  0.3451, 0.3036, 0.2571, 0.2051, 0.1459, \\
& 0.6406 &  0.0788 \\
 \hline
12 & 0.1958, 0.3708, 0.3875, 0.4103, 0.4185, &  0.6293, 0.5103, 0.4553, 0.4275, 0.3937, \\
& 0.4297, 0.4430, 0.4639, 0.4933, 0.5343, &  0.3612, 0.3248, 0.2849, 0.2406, 0.1913, \\
& 0.5851, 0.6410 &  0.1356, 0.0731 \\
 \hline
13 & 0.1903, 0.3627, 0.3797, 0.4024, 0.4096, &  0.6315, 0.5130, 0.4593, 0.4340, 0.4028, \\
& 0.4191, 0.4290, 0.4450, 0.4668, 0.4975, &  0.3740, 0.3417, 0.3068, 0.2684, 0.2260, \\
& 0.5385, 0.5878, 0.6414 &  0.1792, 0.1266, 0.0681 \\
 \hline
14 & 0.1855, 0.3555, 0.3728, 0.3954, 0.4020, &  0.6334, 0.5152, 0.4627, 0.4392, 0.4103, \\
& 0.4103, 0.4179, 0.4304, 0.4471, 0.4703, &  0.3843, 0.3554, 0.3243, 0.2906, 0.2535, \\
& 0.5017, 0.5425, 0.5902, 0.6418 &  0.2131, 0.1685, 0.1188, 0.0638 \\
 \hline
15 & 0.1811, 0.3489, 0.3667, 0.3893, 0.3954, &  0.6349, 0.5169, 0.4655, 0.4434, 0.4163, \\
& 0.4028, 0.4088, 0.4189, 0.4318, 0.4501, &  0.3927, 0.3664, 0.3387, 0.3086, 0.2758, \\
& 0.4740, 0.5058, 0.5462, 0.5924, 0.6422 &  0.2402, 0.2015, 0.1589, 0.1118, 0.0600 \\
 \hline
16 & 0.1771, 0.3430, 0.3612, 0.3838, 0.3896, &  0.6363, 0.5184, 0.4678, 0.4469, 0.4213, \\
& 0.3964, 0.4011, 0.4095, 0.4197, 0.4343, &  0.3996, 0.3756, 0.3505, 0.3234, 0.2940, \\
& 0.4532, 0.4778, 0.5099, 0.5497, 0.5944, &  0.2624, 0.2281, 0.1910, 0.1504, 0.1056, \\
& 0.6425 &  0.0566 \\
 \hline
17 & 0.1735, 0.3376, 0.3562, 0.3789, 0.3844, &  0.6375, 0.5197, 0.4697, 0.4499, 0.4255, \\
& 0.3907, 0.3946, 0.4016, 0.4099, 0.4217, &  0.4054, 0.3832, 0.3603, 0.3358, 0.3092, \\
& 0.4370, 0.4565, 0.4816, 0.5138, 0.5530, &  0.2807, 0.2501, 0.2171, 0.1816, 0.1426, \\
& 0.5962, 0.6429 &  0.1001, 0.0536 \\
 \hline
\end{tabular}
\caption{Optimal values of $\gammav$ and $\betav$ up to $p=17$.
}
\label{table:optimal_params}
\end{table}

\clearpage 
\begin{table}[tbh]
\centering
\makegapedcells
\begin{tabular}{c|c|c}
$p$ & $\gammav$ &  $\betav$ \\
\hline \hline
18 & 0.1694, 0.3318, 0.3513, 0.3745, 0.3795, &  0.6412, 0.5232, 0.4726, 0.4533, 0.4295, \\
& 0.3858, 0.3886, 0.3943, 0.4007, 0.4080, &  0.4104, 0.3895, 0.3683, 0.3457, 0.3213, \\
& 0.4201, 0.4410, 0.4591, 0.4849, 0.5178, &  0.2956, 0.2680, 0.2387, 0.2073, 0.1730, \\
& 0.5579, 0.5999, 0.6434 &  0.1369, 0.0949, 0.0510 \\
 \hline
19 & 0.1662, 0.3270, 0.3470, 0.3704, 0.3752, &  0.6425, 0.5245, 0.4743, 0.4556, 0.4327, \\
& 0.3814, 0.3835, 0.3882, 0.3931, 0.3972, &  0.4147, 0.3949, 0.3752, 0.3543, 0.3318, \\
& 0.4064, 0.4260, 0.4517, 0.4620, 0.4885, &  0.3082, 0.2831, 0.2566, 0.2287, 0.1983, \\
& 0.5219, 0.5619, 0.6025, 0.6438 &  0.1653, 0.1307, 0.0903, 0.0486 \\
 \hline
20 & 0.1632, 0.3224, 0.3430, 0.3666, 0.3714, &  0.6438, 0.5258, 0.4758, 0.4577, 0.4355, \\
& 0.3775, 0.3789, 0.3828, 0.3865, 0.3875, &  0.4184, 0.3996, 0.3812, 0.3616, 0.3407, \\
& 0.3942, 0.4129, 0.4376, 0.4541, 0.4649, &  0.3189, 0.2958, 0.2716, 0.2466, 0.2194, \\
& 0.4921, 0.5259, 0.5659, 0.6051, 0.6442 &  0.1899, 0.1581, 0.1251, 0.0862, 0.0464 \\
 \hline
\end{tabular}
\caption{Suboptimal values of $\gammav$ and $\betav$ for $p=18, 19, 20$ based on extrapolations of the optimal values of $\gammav$ and $\betav$ up to $p=17$ of Table~\ref{table:optimal_params}.
The values presented here are used to compute the lower bounds on $\bar{\nu}_p$ of Table~\ref{table:lower_bound_values}.
}
\label{table:lower_bound_params}
\end{table}

\section{Optimal $\gammav$ and $\betav$ for Max-$q$-XORSAT at $p=14$}
\label{apx:optimal_angles_q}
\begin{figure}[H]
    \centering
    \includegraphics[width=\linewidth]{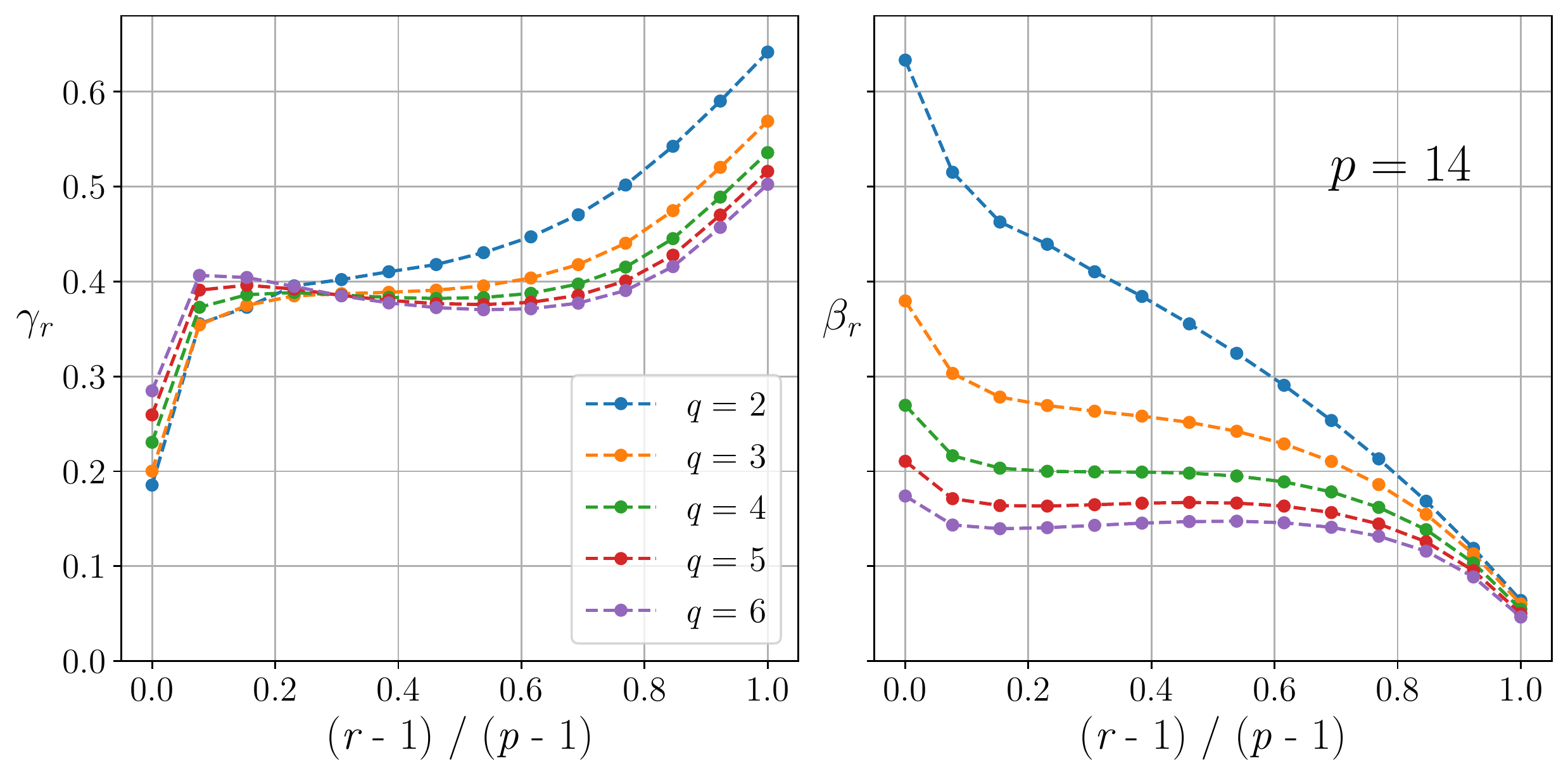}
    \caption{Optimal QAOA parameters $(\paramv)$ at $p=14$ for various Max-$q$-XORSAT on $D$-regular hypergraphs in the $D\to\infty$ limit.
    This data can be found in Ref.~\cite{data}.}
    \label{fig:params_varying_q}
\end{figure}

\end{document}